\documentclass{scrartcl}
\usepackage[utf8]{inputenc}
\usepackage{microtype}
\usepackage[bibliography=common]{apxproof}
\renewcommand{\appendixprelim}{}
\usepackage{graphicx}
\usepackage{booktabs} 
\usepackage{amssymb}
\usepackage{stmaryrd}
\usepackage{paralist}
\usepackage{amsmath}
\usepackage{mathrsfs}
\usepackage{xspace}
\usepackage{url}
\usepackage{bm}
\usepackage[table]{xcolor}
\usepackage{tikz}
\usepackage{float}
\usepackage{colonequals}
\usepackage{hhline}
\usepackage{verbatim}
\usepackage{latexsym}
\usepackage[breaklinks=true]{hyperref}
\usetikzlibrary{arrows.meta}

\hypersetup{
    colorlinks,
    linkcolor={red!50!black},
    citecolor={blue!50!black},
    urlcolor={blue!30!black}
}

\renewcommand{\Pr}{\mathrm{Pr}}

\makeatletter
\newcommand*{\defeq}{\mathrel{\rlap{%
  \raisebox{0.3ex}{$\m@th\cdot$}}%
  \raisebox{-0.3ex}{$\m@th\cdot$}}%
  =}
\makeatother

\newcommand{\card}[1]{\left|#1\right|}

\newcommand{\calH}{\mathcal{H}}
\newcommand{\calG}{\mathcal{G}}

\newcommand{\calX}{\mathcal{X}}
\newcommand{\calY}{\mathcal{Y}}

\newcommand\la{\langle}
\newcommand\ra{\rangle}

\newcommand\restr[2]{{
  \kern-\nulldelimiterspace 
  #1 
  _{|#2} 
  }}

\newcommand{\QQp}{\mathbb{Q}_+}
\newcommand{\ZZ}{\mathbb{Z}}

\newcommand{\class}[1]{\ensuremath{\mathsf{#1}}}
\newcommand{\DWT}{\class{DWT}\xspace}
\newcommand{\PT}{\class{PT}\xspace}
\newcommand{\OWP}{\class{1WP}\xspace}
\newcommand{\TWP}{\class{2WP}\xspace}
\newcommand{\Connected}{\class{Connected}\xspace}
\newcommand{\All}{\class{All}\xspace}

\newcommand{\UPT}{\class{\bigsqcup PT}\xspace}
\newcommand{\UDWT}{\class{\bigsqcup DWT}\xspace}
\newcommand{\UOWP}{\class{\bigsqcup1WP}\xspace}
\newcommand{\UTWP}{\class{\bigsqcup2WP}\xspace}

\newcommand{\problem}[1]{\texttt{\mdseries \upshape #1}}
\newcommand{\cspd}{\ensuremath{\problem{\#CSP}_\mathrm{d}}}
\newcommand{\pptdnf}{\problem{\#PP2DNF}}

\newcommand{\regline}{\arrayrulecolor{black}}
\newcommand{\hardline}{\arrayrulecolor{black!20!white}}
\newcommand{\hardcell}{\cellcolor{black!20!white}}
\newcommand{\tabexplanation}{{%
  \setlength{\fboxsep}{1pt}%
  \footnotesize
  {\fcolorbox{black}{white}{\vphantom{\#}PTIME}}
  {\fcolorbox{black}{black!20!white}{\#P-hard}}
  Numbers given correspond to propositions for border
  cases, remaining cells can be filled using
  the inclusions from Figure~\ref{fig:inclusion_diagram}.%
}}

\newcommand\phom{\ensuremath{\mathop{\mathsf{PHom}}}\xspace}

\newcommand\phomL{\ensuremath{\mathop{\mathsf{PHom}_{\mathsf{L}}}}\xspace}
\newcommand\phomU{\ensuremath{\mathop{\mathsf{PHom}_{\!\not\:\mathsf{L}}}}\xspace}

\newcommand\phoml[2]{\ensuremath{\mathop{\mathsf{PHom}_{\mathsf{L}}}(#1,#2)}}
\newcommand\phomu[2]{\ensuremath{\mathop{\mathsf{PHom}_{\!\not\:\mathsf{L}}}(#1,#2)}}

\renewcommand{\hom}{\leadsto}

\newcommand{\vlineright}[1]{\multicolumn{1}{c|}{#1}}

 \newtheorem{definition}{Definition}[section]
 
\newtheoremrep{proposition}[definition]{Proposition}
\newtheoremrep{theorem}[definition]{Theorem}
 
 \newtheoremrep{lemma}[definition]{Lemma}
 \newtheorem{example}[definition]{Example}

\title{Conjunctive Queries on Probabilistic Graphs:\\
Combined Complexity}

\date{}

\author{
\begin{tabular}[t]{c}
Antoine Amarilli \\
{\normalfont LTCI, Télécom ParisTech, Université Paris-Saclay} \\
{\normalfont antoine.amarilli@telecom-paristech.fr} \\[0.5em]
Mikaël Monet \\
{\normalfont LTCI, Télécom ParisTech, Université Paris-Saclay} \\
{\normalfont mikael.monet@telecom-paristech.fr} \\[0.5em]
Pierre Senellart \\
{\normalfont DI, École normale supérieure, PSL Research University} \\
{\normalfont \& Inria Paris} \\
{\normalfont pierre.senellart@telecom-paristech.fr} \\[0.5em]
\end{tabular}
}

\renewcommand{\phi}{\varphi}
\renewcommand{\epsilon}{\varepsilon}
\renewcommand{\leq}{\leqslant}

\begin{document}

\maketitle

\begin{abstract}
  Query evaluation over probabilistic databases is known to be intractable in many
cases, even in data complexity, i.e., when the query is fixed. Although some
restrictions of the queries \cite{dalvi2012dichotomy} and instances
\cite{amarilli2015provenance} have been proposed to lower the complexity, these
known tractable cases usually do not apply to combined complexity, i.e., when
the query is not fixed. This leaves open the question of which query and
instance languages ensure the tractability of probabilistic query evaluation in
combined complexity.

This paper proposes the first general study of the combined complexity of
conjunctive query evaluation on probabilistic instances over binary signatures,
which we can alternatively phrase as a probabilistic version of the graph
homomorphism problem, or of a constraint satisfaction problem (CSP) variant. We
study the complexity of this problem depending on whether instances and queries
can use features such as edge labels, disconnectedness, branching, and edges in
both directions. We show that the complexity landscape is surprisingly rich,
using a variety of technical tools: automata-based compilation to d-DNNF
lineages as in~\cite{amarilli2015provenance}, $\beta$-acyclic lineages
using~\cite{brault2015understanding}, the $\underline{X}$-property for tractable
CSP from~\cite{gutjahr1992polynomial}, graded DAGs~\cite{odagiri2014greatest}
and various coding techniques for hardness proofs.

\end{abstract}

\section{Introduction}\label{sec:introduction}
Uncertainty naturally arises in many data management applications, when
integrating data that may be untrustworthy, erroneous, or outdated; or when
generating or annotating data using information extraction or machine learning
approaches.
The framework of \emph{probabilistic databases} \cite{suciu2011probabilistic} 
has been introduced to answer such needs: it provides a natural semantics for
concise representations of probability distributions on data, and allows the
user to evaluate queries directly on the representations.
The simplest probabilistic framework is that of \emph{tuple-independent databases}~(TID), where
each tuple in the relational database is annotated with a probability of actually
being present, assuming independence across all tuples. Evaluating a Boolean query $Q$
over a TID instance $I$ means computing the probability that $Q$ is true
according to the distribution of~$I$, or in other words, the total probability mass of the
possible worlds of~$I$ that satisfy~$Q$.

As is usual in database theory, the complexity of this probabilistic query
evaluation problem (PQE) can be measured as a function of both $I$ and $Q$,
namely, \emph{combined complexity}~\cite{vardi1995complexity}, or as a function of $I$ when the query~$Q$
is fixed, called \emph{data complexity}. Almost all works on PQE so far
have focused on data complexity, where they have explored the general
intractability of PQE in this sense. Indeed, while non-probabilistic query evaluation 
of fixed queries in first-order logic
has
polynomial-time data complexity (specifically, $\text{AC}^0$), the PQE problem
is \#P-hard\footnote{\#P~is the class of counting
problems that can be expressed as the number of accepting paths of a
nondeterministic polynomial-time Turing machine.} already for some fixed
conjunctive queries~\cite{dalvi2007efficient}. Specifically,
the celebrated PQE result
by Dalvi and Suciu~\cite{dalvi2012dichotomy},
has shown a dichotomy on unions of conjunctive queries: some are 
\emph{safe queries}, enjoying PTIME data complexity (specifically, linear~\cite{ceylan2016open}),
and all other queries are \#P-hard. Earlier work by some
of the present authors has shown also a dichotomy on instance families for fixed
monadic second-order queries, with tractable data complexity for bounded-treewidth
families \cite{amarilli2015provenance}, and intractability otherwise under some
assumptions~\cite{amarilli2016tractable}.

However, even when PQE is tractable in data complexity, the task may still be
infeasible because of unrealistically large constants that depend on the query.
For instance, our approach in~\cite{amarilli2015provenance} is nonelementary in
the query, and the algorithm for safe queries in~\cite{dalvi2012dichotomy} is
generally super-exponential in the query~\cite{suciu2011probabilistic}.
For this reason, we believe that it is also important to achieve a good
understanding of the \emph{combined} complexity of PQE, and to isolate cases where PQE
is tractable in combined complexity; similarly to how, e.g., Yannakakis'
algorithm can evaluate $\alpha$-acyclic queries on non-probabilistic instances 
with tractable combined complexity~\cite{yannakakis1981algorithms}.
This motivates the question studied in this paper: \emph{for which
classes of queries and instances does PQE enjoy tractable combined complexity?}

\paragraph*{Related work.}
Surprisingly, the question of achieving combined tractability for PQE
does not seem to have been studied before. To our knowledge, the only exception
is in the setting of probabilistic XML~\cite{kimelfeld2013probabilistic}, where
deterministic tree automata queries were shown to enjoy
tractable combined complexity~\cite{cohen2009running}. In the context of
relational databases, our recent work~\cite{amarilli2017combined} shows the
combined tractability of \emph{provenance computation} for a specific Datalog
fragment on bounded-treewidth instances, but observes that these results do not
seem to give tractability of PQE, which is already intractable in much more
restricted settings. These results, however (Propositions~36 and~38
of~\cite{amarilli2017combinedb}), do not give a complete picture of the combined
complexity of PQE; in particular, they do not even give any non-trivial setting
where it is tractable.

Questions of combined tractability have also been studied in the setting of
\emph{constraint satisfaction problems} (CSP), following a well-known connection
between CSP and the conjunctive query evaluation problem in database
theory, or the study of the graph homomorphism problem (see, e.g.,
\cite{grohe2007complexity}). We can then see the restriction of PQE to conjunctive queries as a probabilistic, or weighted,
variant of these problems, but we are not aware of any existing study of this
variant. In the graph homomorphism setting, a related but
different problem is that of \emph{counting} graph
homomorphisms~\cite{bulatov2013complexity}: but this
amounts to counting the number of matches of a query in a database instance,
which is different from counting the possible worlds of an instance where the
query has some match, as we do.
A more related problem is \problem{\#SUB}~\cite{curticapean2014complexity},
which asks, given a query graph~$G$ and an instance graph~$H$, for the
\emph{number of subgraphs} of~$H$ which are \emph{isomorphic} to~$G$. When all facts
are labeled with $1/2$, our problem asks instead for the number of subgraphs of~$H$ to
which $G$ admits a \emph{homomorphism}. A~further difference is that we allow arbitrary
probability annotations, amounting to a form of weighted counting; in
particular,
facts can be given probability~1.

\paragraph*{Problem statement.}
Inspired by the connection to graph homomorphism and CSP, in this paper we investigate
the probabilistic query evaluation problem for conjunctive queries on
tuple-independent instances, over arity-two
signatures. 
To our knowledge, our paper is the first to focus on the combined complexity of
conjunctive query evaluation on probabilistic relational data.
For simplicity of exposition, we will phrase our problem in terms of graphs: given a
\emph{query graph} and a probabilistic \emph{instance graph}, where each edge is
annotated by a probability, we must determine the probability that the query graph has a
homomorphism to the instance graph, i.e., the total probability mass of the
subgraphs which ensure this, assuming independence between edges. We always
assume the query and instance graphs to be directed.

As we will see, the problem is generally intractable, so we will have to study
restricted settings. We accordingly study this problem under assumptions on the
query and input graphs.
Inspired by our prior intractability results~\cite{amarilli2016tractable},
one general assumption that we will make is to impose \emph{tree-likeness} of the
instance. In fact, we will generally restrict it to be a \emph{polytree}, i.e., a directed
graph whose underlying undirected graph is a tree. As we will see, however, even
this restriction does not suffice to ensure tractability, so we study the impact
of several other features:

\begin{itemize}
\item \emph{Labels}, i.e., whether edges of the query and instance can be labeled
  by a finite alphabet, as would be the case on a relational signature with
  more than one binary predicate.
\item \emph{Disconnectedness}, i.e., allowing disconnected queries and
  instances.
\item \emph{Branching}, i.e., allowing graphs to branch out, instead of
requiring them to be a path.
\item \emph{Two-wayness}, i.e., allowing edges with arbitrary orientation,
  instead of requiring all edges to have the same orientation (as in a one-way
  path, or downward tree).
\end{itemize}
We accordingly study our problem for \emph{labeled graphs} and \emph{unlabeled
graphs}, and when query and instance graphs are in the following classes, that cover
the possible combinations of the above characteristics: one-way and two-way
paths, downward trees and polytrees, and disjoint unions thereof.

\paragraph*{Results.}
This paper presents our combined complexity results for the probabilistic query evaluation problem in all
these settings. After introducing the preliminaries and defining the problem in 
Section~\ref{sec:preliminaries}, we first study the impact of disconnectedness
in instances and queries in Section~\ref{sec:disconnected-CQs}. While we can
easily show that disconnectedness does not matter for instances
(Lemma~\ref{lem:disconnected-instances}), we show that
disconnectedness of queries has an unexpected impact on complexity: in the
labeled case, even the simplest disconnected queries on the simplest
kinds of instances are intractable (Proposition~\ref{prp:l-conj_of_1WP-1WP}): this result is shown
via the hardness of counting edge covers in bipartite graphs.
The picture for
disconnected queries is more complex in the unlabeled case (see
Table~\ref{tab:unlabeled-conj_of}): indeed, the problem is still hard when allowing
two-wayness in the query and instance (as it can be used to simulate labels, see
Proposition~\ref{prp:u-conj_of_2WP-2WP}), but disallowing two-wayness in the
instance ensures tractability of all queries. This latter result
(Proposition~\ref{prp:u-all-DWT}) is established by showing that all queries
then
essentially collapse to a one-way path: we do so by assigning a \emph{level} to all
vertices of the query using a notion of \emph{graded
DAGs}~\cite{odagiri2014greatest,schroder2016graded}.

We then focus on connected queries, and first study the labeled setting in
Section~\ref{sec:l-connected-CQs}; see Table~\ref{tab:labeled} for a summary of
results. We show that disallowing instance branching ensures the tractability of
all connected queries (Proposition~\ref{prp:l-connected-2WP}), and that
disallowing branching in the query \emph{and} two-wayness in the instance and query
also does (Proposition~\ref{prp:l-1WP-DWT}). These two results are shown by
computing a \emph{Boolean lineage} of the query~\cite{suciu2011probabilistic},
and proving that we can tractably evaluate its probability because it is
$\beta$-acyclic~\cite{brault2015understanding}, thanks to the restricted
instance structure. For the first result, this process further relies on a CSP tool
to show the tractability of homomorphism testing in labeled two-way paths, a
condition dubbed the $\underline{X}$-property
\cite{gutjahr1992polynomial,gottlob2006conjunctive}. We show the intractability
of all other cases (Propositions~\ref{prp:l-1WP-PT}, \ref{prp:l-DWT-DWT}, and
\ref{prp:l-2WP-DWT}), by coding \#SAT-reduction, reusing in part a coding from
\cite{amarilli2017combinedb}.

We last study the unlabeled setting for connected queries in
Section~\ref{sec:u-connected_CQs}. We show that disallowing query branching and
two-wayness suffices to obtain tractability, provided that the instance is a
polytree (Proposition~\ref{prp:u-1WP-PT}): this result is proven by building in
PTIME a deterministic tree automaton to test the length of the longest path, and
compiling a d-DNNF lineage as in \cite{amarilli2015provenance}. This result
immediately extends to branching queries, as they are equivalent to paths in
this case (Proposition~\ref{prp:u-DWT-PT}). We complete the picture by showing
that, by contrast, allowing two-wayness in the query leads to intractability on
polytrees, by a variant of our coding technique
(Proposition~\ref{prp:u-2WP-PT}). We then conclude in
Section~\ref{sec:conclusion}.

Our results completely classify the complexity
of probabilistic conjunctive query evaluation for all combinations of instance
and query restrictions, in the labeled and unlabeled setting.
Full proofs are given in appendix.

\section{Preliminaries}\label{sec:preliminaries}
We first provide some formal definitions of the concepts we use in this
paper, and introduce the \emph{probabilistic graph homomorphism} problem
and the different classes of graphs that we consider.

\paragraph*{Graphs and homomorphisms.}
Let $\sigma$ be a finite non-empty set of labels. When $|\sigma|>1$, we
say that we are in the \emph{labeled setting}; when $|\sigma|=1$, in
the \emph{unlabeled setting}.

We consider \emph{directed graphs} with edge labels from~$\sigma$, i.e.,
triples $H = (V,E,\lambda)$ with $V$ a non-empty finite set of vertices, $E\subseteq
V^2$ a set of edges, and $\lambda: E\to\sigma$ a labeling function.
We write $a\xrightarrow{R} b$ for an edge $e=(a,b)$ with label
$\lambda(e)=R$.
Note that we do not allow multi-edges: an edge $e$ has a unique label $\lambda(e)$.
When $|\sigma| = 1$, i.e., in the unlabeled setting, we simply write 
$(V,E)$ for the graph and $a \rightarrow b$ for an edge.
Unless otherwise specified, all graphs that we consider
in this paper are directed.

A graph $H'=(V',E',\lambda')$ is a \emph{subgraph} of the
graph~$H=(V,E,\lambda)$, written $H'\subseteq H$, when we have $V'=V$,
$E'\subseteq E$, and when
$\lambda'$ is $\lambda_{|E'}$, i.e., the restriction of~$\lambda$ to~$E'$.
(Note that, in a slightly non-standard way, we impose that
subgraphs have the same set of vertices than the original graph; this
will simplify some notation.)

A \emph{graph homomorphism} $h$ from some graph~$G=(V_G,E_G,\lambda_G)$
to some
graph~$H=(V_H,E_H,\lambda_H)$ is a function $h: V_G\to V_H$ such that, for all
$(u,v)\in E_G$, we have $(h(u),h(v))\in E_H$ and further $\lambda_H((h(u),
h(v))) = \lambda_G((u, v))$.
A \emph{match} of~$G$ in~$H$ is the image in~$H$ of such a homomorphism~$h$, i.e.,
the graph with vertices $h(u)$ for~$u \in V_G$ and edges $(h(u),
h(v)))$ for~$(u, v) \in E_G$. Note that two different homomorphisms may define
the same match.
Also note that two distinct nodes of $G$ could have the same image by $h$, so a match
of~$G$ in~$H$ is not necessarily homomorphic to~$G$.
We write $G\hom H$ when there exists a homomorphism from~$G$ to~$H$. We call
two graphs $G$ and~$G'$ \emph{equivalent} if, for any graph $H$, we have $G \hom
H$ iff $G' \hom H$. It is easily seen that $G$ and $G'$ are equivalent if and only if $G \leadsto G'$ and $G' \leadsto G$.

\paragraph*{Probabilistic graphs.}
A \emph{probability distribution on graphs} is a function $\Pr$ from a finite set
$\mathcal{W}$ of graphs (called the \emph{possible worlds} of $\Pr$) to values
in $[0;1]$
represented as rational numbers, such that the probabilities of all possible
worlds sum to~$1$, namely, $\sum_{H \in \mathcal{W}} \Pr(H) = 1$.

A \emph{probabilistic graph} is intuitively a concise representation of a probability
distribution. Formally, it is a pair $(H,\pi)$ where $H$ is a 
graph with edge labels from~$\sigma$ and where $\pi$ is a probability function
$\pi: E\to [0;1]$ that maps every edge $e$ of $H$ to a probability~$\pi(e)$,
represented as a rational number. Note that each edge $(u, v)$
in a probabilistic graph~$(H, \pi)$
is annotated both with a label $\lambda((u, v)) \in \sigma$, and a
probability $\pi((u, v))$. 

The \emph{probability distribution} $\Pr$
defined by the probabilistic graph $(H,\pi)$ is obtained intuitively by
considering that edges are kept or deleted independently according to the
indicated probability. Formally, the possible worlds $\mathcal{W}$ of~$\Pr$ are
the subgraphs of~$H=(V,E,\lambda)$, and for $H'=(V,E',\lambda_{|E'}) \subseteq H$ we define $\Pr(H') \defeq
\prod_{e \in E'} \pi(e) \times \prod_{e \in E \setminus E'} (1-\pi(e))$. Note
that, when $H$ has edges labeled with~$0$ or~$1$, some possible worlds are given
probability~$0$ by $\pi$.

\begin{example}
  \label{exa:graph}
  Figure~\ref{fig:full-example} represents a probabilistic graph $(H, \pi)$ on
  signature $\sigma = \{R, S\}$, where each edge is annotated with its label and
  probability value. There are $2^6$ possible worlds, $2^5$ of which have
  non-zero probability. 

  The possible world where all $R$-edges are kept and all $S$-edges are removed
has probability $0.1 \times 1 \times 0.8 \times 0.1 \times 0.05 \times (1 - 0.7)$.
\end{example}

\paragraph*{Probabilistic graph homomorphism.}
The goal of this paper is to study the
\emph{probabilistic homomorphism problem} $\phom$, for the set of labels
$\sigma$ that we fixed: given a
graph $G$ on~$\sigma$ and a probabilistic graph~$(H,\pi)$ on~$\sigma$, compute the probability
that there exists a homomorphism from $G$ to $H$ under $\Pr$, i.e.,
the sum of the probabilities of all subgraphs of~$H'$ to which $G$ has a
homomorphism:
\[
  \Pr(G\leadsto H)\defeq \sum_{\substack{H'\subseteq H\\G\leadsto H'}}\Pr(H').
\]

\begin{example}
  Continuing the example, let us consider the $\phom$ problem for the graph
  ${G:
  {\xrightarrow{R} \xrightarrow{S} \xleftarrow{S}}}$ and the example probabilistic
  graph $(H, \pi)$ in Figure~\ref{fig:full-example}. 
  The graph $G$ intuitively corresponds to the relational calculus query $\exists x y z t ~ R(x,y) \land S(y,z)
  \land S(t,z)$.
  Of course, we can compute $\Pr(G \leadsto H)$ by summing over the possible
  worlds of~$H$, but this process is generally intractable.
  Here, by considering the possible matches of $G$ in $H$, we can
  see that $\Pr(G \leadsto H) = 0.7 \times (1 - (1 - 0.1) \times (1 -
  0.8))$. 
\end{example}

\begin{figure}
  \centering
  \begin{tikzpicture}[yscale=.4,xscale=1.5]
	\node[fill = white, name = H] at (-.5,2) {$H$:};

	\node[fill = white, name = 20] at (2,0) {$\bullet$};
	\node[fill = white, name = 42] at (4,2) {$\bullet$};
	\node[fill = white, name = 24] at (2,4) {$\bullet$};
	\node[fill = white, name = 02] at (0,2) {$\bullet$};
	\draw[line width = 1pt, ->]  (20) -- (02) node[midway,pos=.33,below=1.5pt] {$1$}
        node[pos=0.67,below=1.5pt] {$R$};
	\draw[line width = 1pt, ->]  (02) -- (24)
        node[midway,pos=.67,above=1.5pt] {$0.1$}
        node[pos=0.33,above=1.5pt] {$R$};
	\draw[line width = 1pt, ->]  (24) -- (42) node[midway,pos=.33,above=1.5pt] {$R$}
        node[pos=0.67,above=1.5pt] {$0.1$};
	\draw[line width = 1pt, ->]  (42) -- (20) node[midway,pos=.67,below=1.5pt] {$R$}
        node[pos=0.33,below=1.5pt] {$0.05$};

        \draw[line width = 1pt, ->]  ([xshift=1.7pt]24.south) to 
        node[right,midway,align=left] {$S$\\$0.7$} ([xshift=1.7pt]20.north);
        \draw[line width = 1pt, ->]  ([xshift=-1.7pt]20.north) to 
        node[left,midway,align=right]
        {$R$\\$0.8$}
        ([xshift=-1.7pt]24.south);

\end{tikzpicture}
  \caption{Example probabilistic graph $H$}
\label{fig:full-example}
\end{figure}
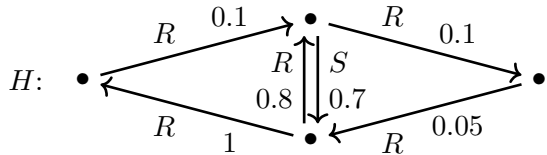

Following database terminology, we call $G$ the \emph{query graph} and $(H,
\pi)$ the \emph{(probabilistic) instance graph}. Indeed, the $\phom$ problem is
easily seen to be equivalent to conjunctive query evaluation on probabilistic
tuple-independent relational databases \cite{dalvi2012dichotomy},
over binary relational signatures.

Note that, in this paper, we measure the complexity of $\phom$ as a function
of \emph{both} the query graph $G$ and of the instance graph $(H, \pi)$, i.e.,
in database terminology, we measure the \emph{combined
complexity}~\cite{vardi1995complexity} of
probabilistic query evaluation. The $\phom$ problem is known to be
\mbox{\#P-hard} in
general \cite{dalvi2007efficient} (even for some fixed query graphs): by this,
we mean that it is hard (under polynomial-time reductions) for the class \#P of
counting problems that can be expressed as the number of accepting paths of a
nondeterministic polynomial-time Turing machine.
To achieve tractable complexity for $\phom$, we will classify the complexity of $\phom$
under various restrictions. We say that the complexity of some variant of the
problem is \emph{tractable} if the probability can be computed by a deterministic
polynomial-time Turing machine: by a slight abuse of terminology, we then say
that it is \emph{in PTIME}.
All $\phom$ variants that we study will be shown either to be PTIME in this
sense, or to be \#P-hard.

We will study restrictions of $\phom$ first by distinguishing the 
\emph{labeled} and \emph{unlabeled} settings. We write
$\phomL$ for the problem when the fixed label set $\sigma$ is such that
$|\sigma| > 1$, and $\phomU$ when the fixed $\sigma$ is such that $|\sigma|=1$.

The second restriction concerns the input query graphs and instance graphs.
We will model
restrictions on these graphs by requiring them to be taken from specific
graph classes,
where by \emph{graph class} we simply mean an infinite set of graphs.
Inspired by the notation used in CSP, for two classes $\mathcal{G}$ and
$\mathcal{H}$ of graphs in the labeled setting,
we denote $\phoml{\mathcal{G}}{\mathcal{H}}$
the problem that takes as input
a graph $G$ in class~$\mathcal{G}$ and a probabilistic graph $(H,\pi)$ with $H$ in class
$\mathcal{H}$,
and computes the probability $\Pr(G\leadsto H)$.
We denote the same problem in the unlabeled setting by 
$\phomu{\mathcal{G}}{\mathcal{H}}$.

\begin{figure}
  \centering
  \begin{tikzpicture}[yscale=.49]
	\node[fill = white, name = 1WP] at (0,0) {\OWP};
	\node[fill = white, name = 2WP] at (1.8,1) {\TWP};
	\node[fill = white, name = DWT] at (1.8,-1) {\DWT};
	\node[fill = white, name = PT] at (3.6,0) {\PT};
	\node[fill = white, name = connected] at (5.6,0) {\Connected};
	\node[fill = white, name = all] at (7.6,0) {\All};
\path[-stealth, auto] (1WP) edge[draw=none] node[sloped, auto=false] {$\subseteq$} (2WP);
\path[-stealth, auto] (2WP) edge[draw=none] node[sloped, auto=false, allow upside down] {$\subseteq$} 
(PT);
 \path[-stealth, auto] (PT) edge[draw=none] node[sloped, auto=false] {$\subseteq$} (connected);
 \path[-stealth, auto] (connected) edge[draw=none] node[sloped, auto=false] {$\subseteq$} (all);
 \path[-stealth, auto] (1WP) edge[draw=none] node[sloped, auto=false] {$\subseteq$} (DWT);
 \path[-stealth, auto] (DWT) edge[draw=none] node[sloped, auto=false] {$\subseteq$} (PT);
\end{tikzpicture}
\caption{Inclusions between classes of graphs}
\label{fig:inclusion_diagram}
\end{figure}

\paragraph*{Graph classes.}
The graph classes which we study in this paper are defined as follows, on a graph
$G$ with edge labels from $\sigma$:
\begin{itemize}
\item $G$ is a \emph{one-way
path} (\OWP) if it is of the form
$a_1 \xrightarrow{R_1} \cdots \xrightarrow{R_{m-1}} a_m$ for some~$m$, with all
$a_1, \ldots, a_m$ being pairwise distinct, and with $R_i \in \sigma$ for $1 \leq i < m$.
\item $G$ is a \emph{two-way path} (\TWP) if it is of the form $a_1 - \cdots -
  a_m$, with all $a_1, \ldots, a_m$ being pairwise distinct, and each $-$ being
  $\xrightarrow{R_i}$ or $\xleftarrow{R_i}$ (but not both) for some label $R_i \in \sigma$.
\item $G$ is a \emph{downwards tree} (\DWT) if it is a rooted unranked
  tree (each node can have an arbitrary number of children), with all edges going from parent to child in the tree.
\item $G$ is a \emph{polytree} (\PT) if its underlying undirected graph is an  unranked tree, without restriction on edge directions.
\end{itemize}
We also consider the class \Connected{} of connected graphs, and write 
\All{} the class of all graphs.
The inclusion diagram between our graph classes is shown in
Figure~\ref{fig:inclusion_diagram}.
See Figure~\ref{fig:prelim-example-paths} for an example
of a labeled one-way path and two-way path, and 
Figure~\ref{fig:prelim-example-trees} for an unlabeled downwards tree and polytree.

\begin{figure}
  \centering
  \begin{tikzpicture}[xscale=1.4,yscale=.41]
	\node[fill = white, name = a1] at (0,6) { };
	\node[fill = white, name = a2] at (1,6) { };
	\node[fill = white, name = a3] at (2,6) { };
	\node[fill = white, name = a4] at (3,6) { };
	\node[fill = white, name = a5] at (4,6) { };
	\node[fill = white, name = b1] at (0,4) { };
	\node[fill = white, name = b2] at (1,4) { };
	\node[fill = white, name = b3] at (2,4) { };
	\node[fill = white, name = b4] at (3,4) { };
	\node[fill = white, name = b5] at (4,4) { };
	\node[fill = white, name = b6] at (5,4) { };

	\draw[line width = 1pt, ->]  (a1) -- (a2) node[midway, above, fill=white] {$R$};
	\draw[line width = 1pt, ->]  (a2) -- (a3) node[midway, above, fill=white] {$S$};
	\draw[line width = 1pt, ->]  (a3) -- (a4) node[midway, above, fill=white] {$S$};
	\draw[line width = 1pt, ->]  (a4) -- (a5) node[midway, above, fill=white] {$T$};
	\draw[line width = 1pt, ->]  (b1) -- (b2) node[midway, above, fill=white] {$R$};
	\draw[line width = 1pt, <-]  (b2) -- (b3) node[midway, above, fill=white] {$S$};
	\draw[line width = 1pt, <-]  (b3) -- (b4) node[midway, above, fill=white] {$S$};
	\draw[line width = 1pt, ->]  (b4) -- (b5) node[midway, above, fill=white] {$T$};
	\draw[line width = 1pt, <-]  (b5) -- (b6) node[midway, above, fill=white] {$R$};
\end{tikzpicture}

\caption{Example of labeled \OWP{} (top) and \TWP{} (bottom)
for $\sigma = \{R,S,T\}$.}
\label{fig:prelim-example-paths}
\end{figure}
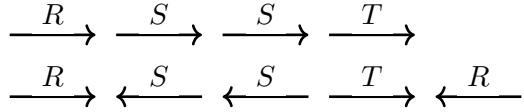

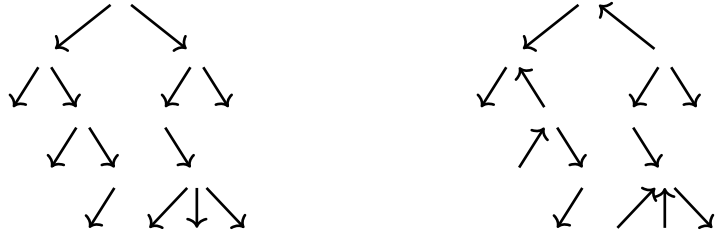
\begin{figure}
  \null\hfill\begin{tikzpicture}[xscale=1,yscale=.8]
	\node[fill = white, name = c1] at (7.5,6.5) { };
	\node[fill = white, name = c2] at (6.5,5.5) { };
	\node[fill = white, name = c3] at (8.5,5.5) { };
	\node[fill = white, name = c4] at (6,4.5) { };
	\node[fill = white, name = c9] at (7.5,3.5) { };
	\node[fill = white, name = c6] at (8,4.5) { };
	\node[fill = white, name = c7] at (9,4.5) { };
	\node[fill = white, name = c8] at (6.5,3.5) { };
	\node[fill = white, name = c5] at (7,4.5) { };
	\node[fill = white, name = c10] at (8.5,3.5) { };
	\node[fill = white, name = c11] at (7,2.5) { };
	\node[fill = white, name = c12] at (7.75,2.5) { };
	\node[fill = white, name = c13] at (8.5,2.5) { };
	\node[fill = white, name = c14] at (9.25,2.5) { };

	\draw[line width = 1pt, ->]  (c2) -- (c4);
	\draw[line width = 1pt, ->]  (c2) -- (c5);
	\draw[line width = 1pt, ->]  (c1) -- (c2);
	\draw[line width = 1pt, ->]  (c1) -- (c3);
	\draw[line width = 1pt, ->]  (c3) -- (c6);
	\draw[line width = 1pt, ->]  (c3) -- (c7);
	\draw[line width = 1pt, ->]  (c5) -- (c8);
	\draw[line width = 1pt, ->]  (c5) -- (c9);
	\draw[line width = 1pt, ->]  (c9) -- (c11);
	\draw[line width = 1pt, ->]  (c6) -- (c10);
	\draw[line width = 1pt, ->]  (c10) -- (c12);
	\draw[line width = 1pt, ->]  (c10) -- (c13);
	\draw[line width = 1pt, ->]  (c10) -- (c14);
  \end{tikzpicture}
  \hfill
  \begin{tikzpicture}[xscale=1,yscale=.8]
	\node[fill = white, name = d1] at (11,6.5) { };
	\node[fill = white, name = d2] at (10,5.5) { };
	\node[fill = white, name = d3] at (12,5.5) { };
	\node[fill = white, name = d4] at (9.5,4.5) { };
	\node[fill = white, name = d9] at (11,3.5) { };
	\node[fill = white, name = d6] at (11.5,4.5) { };
	\node[fill = white, name = d7] at (12.5,4.5) { };
	\node[fill = white, name = d8] at (10,3.5) { };
	\node[fill = white, name = d5] at (10.5,4.5) { };
	\node[fill = white, name = d10] at (12,3.5) { };
	\node[fill = white, name = d11] at (10.5,2.5) { };
	\node[fill = white, name = d12] at (11.25,2.5) { };
	\node[fill = white, name = d13] at (12,2.5) { };
	\node[fill = white, name = d14] at (12.75,2.5) { };

	\draw[line width = 1pt, ->]  (d2) -- (d4);
	\draw[line width = 1pt, <-]  (d2) -- (d5);
	\draw[line width = 1pt, ->]  (d1) -- (d2);
	\draw[line width = 1pt, <-]  (d1) -- (d3);
	\draw[line width = 1pt, ->]  (d3) -- (d6);
	\draw[line width = 1pt, ->]  (d3) -- (d7);
	\draw[line width = 1pt, <-]  (d5) -- (d8);
	\draw[line width = 1pt, ->]  (d5) -- (d9);
	\draw[line width = 1pt, ->]  (d9) -- (d11);
	\draw[line width = 1pt, ->]  (d6) -- (d10);
	\draw[line width = 1pt, <-]  (d10) -- (d12);
	\draw[line width = 1pt, <-]  (d10) -- (d13);
	\draw[line width = 1pt, ->]  (d10) -- (d14);
\end{tikzpicture}\hfill\null
\centering
\caption{Examples of unlabeled \DWT{} (left) and \PT{} (right)}
\label{fig:prelim-example-trees}
\end{figure}

We also introduce the classes \UOWP{} (resp., \UTWP{}, \UDWT{}, \UPT{}) of graphs that are
\emph{disjoint unions of \OWP{} \emph{(resp., \TWP{}, \DWT{}, \PT{})}}, that is, of
possibly disconnected graphs whose connected components are \OWP{} (resp.,
\TWP{}, \DWT{}, \PT{}).

Our graph classes were chosen to be representative of different features
of graphs that will have an impact in the complexity of the \phom{}
problem, namely, \emph{labeling}, \emph{two-wayness}, \emph{branching}, and
\emph{disconnectedness}. Indeed, \TWP{}
(resp., \PT{}) adds two-wayness to \OWP{} (resp., \DWT{}); \DWT{}
(resp., \PT{}) adds branching to \OWP{} (resp., \TWP{}); and \UOWP{}
(resp., \UTWP{}, \UDWT{}, \UPT{}) adds disconnectedness to \OWP{} (resp.,
\TWP{}, \DWT{}, \PT{}).

In the following sections, we investigate the complexity of 
probabilistic graph homomorphism for these various classes of conjunctive
queries and instances.

\section{Disconnected Case}\label{sec:disconnected-CQs}
We first consider the case where either the query or probabilistic
instance graph is \emph{disconnected}, i.e., not in the \Connected{}
class. When the query is disconnected, we show in this section that the
probabilistic homomorphism problem is \#P-hard in all but the most
restricted of cases (in particular in the labeled setting),
which justifies that we restrict to connected queries
in the rest of the paper. On the other hand, we will show that
disconnectedness in the probabilistic instance graph has essentially no
impact on combined complexity.

\subsection{Labeled Disconnected Queries}
\label{sec:labeled-disconnected-queries}

\newcommand{\bipec}{\problem{\#Bipartite-Edge-Cover}}

We establish our main intractability result on disconnected queries 
by reduction
from the \bipec{} problem on \emph{undirected} graphs:

\begin{definition}
  An undirected graph is \emph{bipartite} if its vertices can be partitioned
  into two classes such that no edge connects two vertices of the same class. An
  \emph{edge cover} of an undirected graph is a subset of its edges such that
  every vertex is incident to at least one edge of the subset.
  \bipec{} is the problem, given a bipartite undirected
  graph, of counting its number of edge covers.
\end{definition}

This problem was shown in~\cite{khanna2011queries} to be intractable. The result
can also be proven using Valiant's
holographic reductions~\cite{valiant2008holographic} and the results of Cai, Lu,
and Xia~\cite{cai2012holographic}: see Appendix~\ref{apx:holographic}.

\begin{theorem}
  \label{thm:bipec}
  \cite{khanna2011queries,cai2012holographic}
  The \bipec{} problem is \#P-complete.
\end{theorem}

We can then use this result to show intractability for the simplest forms of disconnected
query graphs
(\UOWP) on the simplest forms of probabilistic instance graphs (\OWP), in the
\emph{labeled} case:
\begin{proposition}
\label{prp:l-conj_of_1WP-1WP}
\phoml{\UOWP}{\OWP} is \#P-hard.
\end{proposition}
\begin{inlineproof}
	We reduce from \bipec.
	Let $\Gamma = (X \sqcup Y, E)$ be an input to \bipec, i.e., a
        bipartite undirected graph with parts $X$ and~$Y$;
        we write $X = (x_1,\ldots,x_{n_{\mathrm{l}}})$, $Y =
        (y_1,\ldots,y_{n_{\mathrm{r}}})$, $E
        = (e_1,\ldots,e_m)$, and for all $1 \leq i \leq m$ we write $e_i =
        (x_{l_i}, y_{r_i})$, with $1 \leq l_i \leq n_{\mathrm{l}}$ and $1 \leq
        r_i \leq n_{\mathrm{r}}$.
  \begin{figure*}[t]
\begin{minipage}[c]{0.35\linewidth}  
\begin{tikzpicture}[yscale=0.75,xscale=0.65]
	\node[fill = white] at (3,3.5) {$\Gamma=(X \sqcup Y, E)$:};

	\node[fill = white, name=x1] at (4,5) {$x_1$ };
	\node[fill = white, name=x2] at (4,2) {$x_2$ };

	\node[fill = white, name=y1] at (8,5) {$y_1$ };
	\node[fill = white, name=y2] at (8,3.5) {$y_2$ };
	\node[fill = white, name=y3] at (8,2) {$y_3$ };

	\draw[line width = 1pt] (x1) -- (y1) node[above, pos=.6] {$e_1$};
	\draw[line width = 1pt] (x1) -- (y2) node[above=-1, pos=0.45] {$e_2$};
	\draw[line width = 1pt] (x1) -- (y3) node[below=1, pos=0.7] {$e_3$};

	\draw[line width = 1pt] (x2) -- (y1) node[below=1, pos=0.3] {$e_4$};

\end{tikzpicture}
\end{minipage}
\hfill
\begin{minipage}[c]{0.05\linewidth}  
\begin{tikzpicture}[yscale=0.85,xscale=0.65]
	\node[fill = white] at (13,3.5) {$G$:};
\end{tikzpicture}
\end{minipage}
\begin{minipage}[c]{0.28\linewidth}  
\begin{tikzpicture}[yscale=0.85,xscale=0.65]
	\draw[line width = 1pt, ->]  (14,5.5) -- (15,5.5) node[midway, above] {$C$};
	\draw[line width = 1pt, ->]  (15,5.5) -- (16,5.5) node[midway, above] {$L$};
	\draw[line width = 1pt, ->]  (16,5.5) -- (17,5.5) node[midway, above] {$V$};
	\node[fill = white] at (18,5.5) {$(x_1)$};

	\draw[line width = 1pt, ->]  (14,4.5) -- (15,4.5) node[midway, above] {$C$};
	\draw[line width = 1pt, ->]  (15,4.5) -- (16,4.5) node[midway, above] {$L$};
	\draw[line width = 1pt, ->]  (16,4.5) -- (17,4.5) node[midway, above] {$L$};
	\draw[line width = 1pt, ->]  (17,4.5) -- (18,4.5) node[midway, above] {$V$};
	\node[fill = white] at (19,4.5) {$(x_2)$};
\end{tikzpicture}
\end{minipage}
\begin{minipage}[c]{0.28\linewidth}  
\begin{tikzpicture}[yscale=0.85,xscale=0.65]

	\draw[line width = 1pt, ->]  (14,3.5) -- (15,3.5) node[midway, above] {$V$};
	\draw[line width = 1pt, ->]  (15,3.5) -- (16,3.5) node[midway, above] {$R$};
	\draw[line width = 1pt, ->]  (16,3.5) -- (17,3.5) node[midway, above] {$C$};
	\node[fill = white] at (18,3.5) {$(y_1)$};

	\draw[line width = 1pt, ->]  (14,2.5) -- (15,2.5) node[midway, above] {$V$};
	\draw[line width = 1pt, ->]  (15,2.5) -- (16,2.5) node[midway, above] {$R$};
	\draw[line width = 1pt, ->]  (16,2.5) -- (17,2.5) node[midway, above] {$R$};
	\draw[line width = 1pt, ->]  (17,2.5) -- (18,2.5) node[midway, above] {$C$};
	\node[fill = white] at (19,2.5) {$(y_2)$};

	\draw[line width = 1pt, ->]  (14,1.5) -- (15,1.5) node[midway, above] {$V$};
	\draw[line width = 1pt, ->]  (15,1.5) -- (16,1.5) node[midway, above] {$R$};
	\draw[line width = 1pt, ->]  (16,1.5) -- (17,1.5) node[midway, above] {$R$};
	\draw[line width = 1pt, ->]  (17,1.5) -- (18,1.5) node[midway, above] {$R$};
	\draw[line width = 1pt, ->]  (18,1.5) -- (19,1.5) node[midway, above] {$C$};
	\node[fill = white] at (20,1.5) {$(y_3)$};
\end{tikzpicture}
\end{minipage}
\bigskip

\begin{tikzpicture}[yscale=0.85,xscale=0.65]
	\node[fill = white] at (-1,0) {$H$:};

	\draw[line width = 1pt, ->]  (0,0) -- (1,0) node[midway, above] {$C$};
	\draw[line width = 1pt, ->]  (1,0) -- (2,0) node[midway, above] {$L$};
	\draw[dashed, line width = 1pt, ->]  (2,0) -- (3,0) node[midway, above] {$V$};
	\node[fill = white] at (2.5,-0.5) {$(e_1)$ };
	\draw[line width = 1pt, ->]  (3,0) -- (4,0) node[midway, above] {$R$};
	\draw[line width = 1pt, ->]  (4,0) -- (5,0) node[midway, above] {$C$};
	\draw[line width = 1pt, ->]  (5,0) -- (6,0) node[midway, above] {$L$};
	\draw[dashed, line width = 1pt, ->]  (6,0) -- (7,0) node[midway, above] {$V$};
	\node[fill = white] at (6.5,-0.5) {$(e_2)$ };

	\draw[line width = 1pt, ->]  (7,0) -- (8,0) node[midway, above] {$R$};
	\draw[line width = 1pt, ->]  (8,0) -- (9,0) node[midway, above] {$R$};
	\draw[line width = 1pt, ->]  (9,0) -- (10,0) node[midway, above] {$C$};
	\draw[line width = 1pt, ->]  (10,0) -- (11,0) node[midway, above] {$L$};
	\draw[dashed, line width = 1pt, ->]  (11,0) -- (12,0) node[midway, above] {$V$};
	\node[fill = white] at (11.5,-0.5) {$(e_3)$ };

	\draw[line width = 1pt, ->]  (12,0) -- (13,0) node[midway, above] {$R$};
	\draw[line width = 1pt, ->]  (13,0) -- (14,0) node[midway, above] {$R$};
	\draw[line width = 1pt, ->]  (14,0) -- (15,0) node[midway, above] {$R$};
	\draw[line width = 1pt, ->]  (15,0) -- (16,0) node[midway, above] {$C$};
	\draw[line width = 1pt, ->]  (16,0) -- (17,0) node[midway, above] {$L$};
	\draw[line width = 1pt, ->]  (17,0) -- (18,0) node[midway, above] {$L$};
	\draw[dashed, line width = 1pt, ->]  (18,0) -- (19,0) node[midway, above] {$V$};
	\node[fill = white] at (18.5,-0.5) {$(e_4)$ };

	\draw[line width = 1pt, ->]  (19,0) -- (20,0) node[midway, above] {$R$};
	\draw[line width = 1pt, ->]  (20,0) -- (21,0) node[midway, above] {$C$};
\end{tikzpicture}
\centering
\caption{Illustration of the proof of Proposition~\ref{prp:l-conj_of_1WP-1WP},
for the bipartite graph $\Gamma$. Dashed edges have
probability $\frac{1}{2}$. We show (in parentheses) the edge of $\Gamma$ coded by each $V$-labeled
edge in the instance graph $H$, and the vertex of $\Gamma$ coded by each \OWP{} component of the query
graph $G$.}
\label{fig:reduction_l-conj_of_1WP-1WP}
\end{figure*}
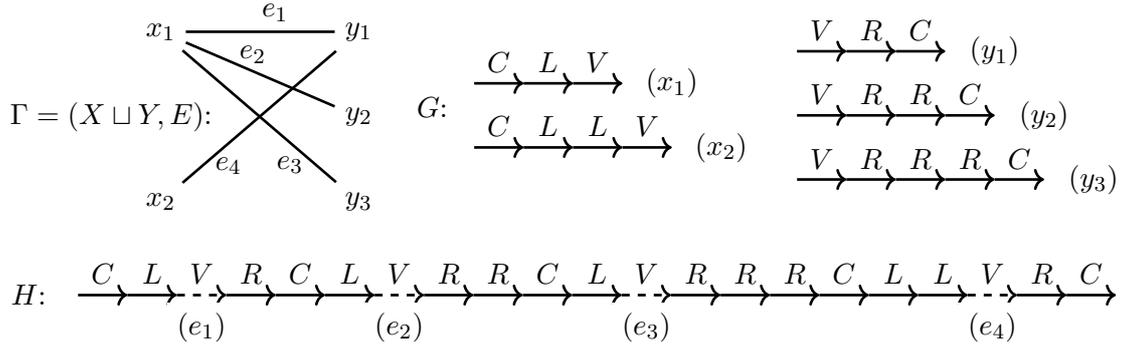

        We first construct in PTIME the \OWP{} probabilistic graph $(H,\pi)$:
        see Figure~\ref{fig:reduction_l-conj_of_1WP-1WP} for an
        illustration of the construction. Specifically,
        for $1 \leq j \leq m$, we construct the following \OWP{}:
        \[H_{e_j} \defeq {(\xrightarrow{L})^{l_j}} \xrightarrow{V}
        (\xrightarrow{R})^{r_j}.\]
        The graph $H$ is then defined as: \[\xrightarrow{C} H_{e_1}
        \xrightarrow{C} H_{e_2} \xrightarrow{C} \cdots \xrightarrow{C} H_{e_m} \xrightarrow{C}.\]
	We define $\pi$ as follows: edges labeled by $V$ have probability
        $\frac{1}{2}$ (intuitively coding whether an edge is part of the
        candidate cover), all others have probability~$1$.

	We then construct the query graph $G \in$ \UOWP, coding the edge covering constraints.
        For every $1 \leq i \leq n_{\mathrm{l}}$, the graph $G$ contains the \OWP{}  component
        $\xrightarrow{C} (\xrightarrow{L})^i \xrightarrow{V}$, and
        for every $1 \leq i \leq n_{\mathrm{r}}$, the graph $G$ contains the \OWP{}  component $\xrightarrow{V} (\xrightarrow{R})^i \xrightarrow{C}$.

        It is clear that $H$ is in \OWP{}, $G$ is in
        \UOWP{} and that both can be constructed in PTIME from $\Gamma$.
        We now show that $\Pr(G\leadsto H)$ is exactly the
        number of edge covers of $\Gamma$ divided by $2^m$, so that the
        computation of the latter reduces in PTIME to the 
	computation of the former, concluding the proof. 

	To see why, we define a bijection between the subsets of edges of
        $\Gamma$, seen as valuations $\nu: E \to \{0,1\}$, to the possible
        worlds $H'$ 
	of $H$ of non-zero probability. We do so in the expected way: keep
        the one $V$-edge $\xrightarrow{V}$ of $H_{e_i}$ iff $\nu(e_i)=1$.
        We now show that there is a homomorphism
	from $G$ to $H'$ if and only if $\nu$ is an edge cover of $\Gamma$.
        As the
        number of $H'$'s such that there is a homomorphism from $G$ to $H'$ is exactly
        $\Pr(G\hom H)\times 2^m$, this will allow us to conclude.
        
        Indeed, if there is a homomorphism $h$ from $G$ to $H'$, then,
        considering the \OWP{} component in $G$ that codes the
        constraint on $x_i$ (resp., on $y_i$),
	its image must be of the form $\xrightarrow{C}
        (\xrightarrow{L})^i \xrightarrow{V}$ (resp., $\xrightarrow{V} (\xrightarrow{R})^i \xrightarrow{C}$),
	but then by construction of $H$
        the $V$-fact must correspond to an edge $e$ such that $x_i$
        (resp., $y_i$) is adjacent to $e$,
        so that we have $\nu(e)=1$ and so $x_i$ (resp., $y_i$) is covered.
        As this is true for each \OWP{} component, all the vertices are covered and $\nu$ is indeed an edge cover of $\Gamma$.
	
	Conversely, suppose that $\nu$ is an edge cover of $\Gamma$, then for
        every vertex $x_i$ (resp., $y_i$) we know that there exists $1 \leq j \leq m$ such that $\nu(e_j)=1$
	and $l_j = i$ (resp., $r_j = i$),
        and we can use the $V$-fact corresponding to $e_j$ and the
        surrounding facts to build the homomorphism as above from each
        component of $G$ to~$H'$.
\end{inlineproof}

The proof of Proposition~\ref{prp:l-conj_of_1WP-1WP} crucially requires
multiple labels in the signature. Indeed, it is easy to see that, in the unlabeled
setting, a query graph in \UOWP{} (or even in \UDWT) is equivalent to
the longest path within the graph, and we will show further
(Proposition~\ref{prp:u-DWT-PT})
that \phomu{\OWP}{\OWP} (indeed, even \phomu{\UDWT}{\PT}) is PTIME.

\subsection{Unlabeled Disconnected Queries}
\label{sec:disconnected-CQs-unlabeled}
In light of this intractability result, let us now consider the unlabeled setting. We show in
Table~\ref{tab:unlabeled-conj_of} where the tractability frontier
lies. First, introducing two-wayness in both query and instance graphs
is enough to obtain an analogue of the intractability of
Proposition~\ref{prp:l-conj_of_1WP-1WP}:

\begin{table}[t]
  \caption{Tractability of \phomU{} for disconnected queries
  (Section~\ref{sec:disconnected-CQs-unlabeled}). Results
 also hold when instances are unions of the indicated classes.}

  \centering
\begin{minipage}{.63\linewidth}
\begin{tabular}{c|ccccc}
  $\downarrow$$G$\qquad $H$$\rightarrow$ & \OWP & \TWP & \DWT & \PT & \Connected \\
  \hline
    \UOWP & & &  & \vlineright{}&
    \hardcell\ref{prp:unlabeled_1wp_on_connected}\\ 
    \hhline{~|~|-|~|-|>{\hardline}-}\regline
    \UTWP & \vlineright{} &
    \vlineright{\hardcell\ref{prp:u-conj_of_2WP-2WP}}&
    \vlineright{}&
    \hardcell & \hardcell \\
    \hhline{~|~|-|~|-|>{\hardline}-}\regline
    \UDWT & & & & \vlineright{\ref{prp:u-DWT-PT}} & \hardcell \\
    \hhline{~|~|-|~|-|>{\hardline}-}\regline
    \UPT & \vlineright{}&
    \vlineright{\hardcell}&
    \vlineright{}&
    \hardcell & \hardcell \\
    \All & \vlineright{} & \vlineright{\hardcell} &
    \vlineright{\ref{prp:u-all-DWT}}&
    \hardcell & \hardcell\\
    \hhline{~|->{\hardline}->{\regline}|-|>{\hardline}--}
\end{tabular}
\label{tab:unlabeled-conj_of}

\tabexplanation
\end{minipage}
\end{table}

\begin{proposition}
\label{prp:u-conj_of_2WP-2WP}
\phomu{\UTWP}{\TWP} is \#P-hard.
\end{proposition}
\begin{inlineproof}
	We reduce, again, from the \#P-hard problem \bipec.
	The idea of the reduction is similar to that used in the proof of Proposition~\ref{prp:l-conj_of_1WP-1WP}, but we face the additional difficulty of not being allowed to use labels.
	Fortunately, we can use two-wayness to simulate them.

        Let $\Gamma=(X\sqcup Y,E)$ be an input of \bipec.
        Consider the reduction from $\Gamma$ used in the proof of
        Proposition~\ref{prp:l-conj_of_1WP-1WP} and the
        \OWP{}  probabilistic graph $(H,\pi)$ and the 
        \UOWP{} query graph $G$ that were constructed.
	We construct from $H$ and $G$ the unlabeled probabilistic graph~$H'$ and
        unlabeled \UTWP{} query graph~$G'$ as follows:
        \smallskip
	\begin{itemize}
		\item replace each $L$- or $R$-labeled edge $a
                  \xrightarrow{L} b$ or $a \xrightarrow{R} b$
			in $H$ and $G$ by $3$ edges $a \rightarrow \rightarrow \leftarrow b$;
		\item replace each $C$-labeled edge $a \xrightarrow{C} b$ of $H$ and $G$ by $3$ edges $a \leftarrow \leftarrow \leftarrow b$;
		\item replace each $V$-labeled edge $a \xrightarrow{V} b$
                  of $H$ and $G$ by $6$ edges $a \rightarrow \rightarrow
                  \rightarrow \rightarrow \rightarrow \leftarrow b$.
	\end{itemize}
        \smallskip
	All edges of $H'$ have probability $1$, except the first edge
        of each sequence of 6 edges that replaced a \mbox{$V$-labeled
        edge},
        which has probability $\frac{1}{2}$.

        Consider a \OWP{}  component of $G$ that codes the constraint on a
        vertex from $Y$, e.g $\xrightarrow{V} (\xrightarrow{R})^i
        \xrightarrow{C}$, which was rewritten in~$G'$ into 
	$\rightarrow \rightarrow \rightarrow \rightarrow \rightarrow \leftarrow 
        {(\rightarrow \rightarrow \leftarrow)^{i}}
        \leftarrow \leftarrow \leftarrow$. A
        homomorphism from this component into a possible world~$J'$ of~$H'$ must
        actually map to a rewriting of a $\xrightarrow{V}
        (\xrightarrow{R})^i \xrightarrow{C}$ sequence in $H'$: indeed,
	the key observation is that the first $5$ $\rightarrow$ edges can only
        be matched to $5$ consecutive $\rightarrow$ in $J'$, which only
        exist as the first 5 edges of a sequence of 6 edges that replaced
        a $V$-labeled fact in $H$. There is no choice left to match the
        subsequent edges without failing. A similar observation holds for 
        components coding the constraints on vertices from~$X$ ($\xrightarrow{C}
        (\xrightarrow{L})^i \xrightarrow{V}$). Hence, we can show correctness of
        the reduction using the same argument as before.
\end{inlineproof}

Allowing two-wayness in both the query and the instance graphs thus allows
us to simulate labels, so that \phomU{} is intractable.
We will study in Section~\ref{sec:u-connected_CQs} what happens
for query graph classes without
two-wayness (i.e., \OWP{}, \DWT{}, and unions thereof); so let us now 
consider the case of instance graph classes where two-wayness is
forbidden, i.e., is in \UDWT. 
As we will show, \phomU{} of
\emph{arbitrary} query graphs on such \UDWT{} instance graphs is tractable.
To this end, we need to introduce \emph{level mappings} of acyclic directed
graphs (DAGs):

\begin{definition}
  \label{def:levelmap}
A \emph{level mapping} of a DAG~$G$ is
a mapping $\mu$ from the vertices of~$G$ to~$\ZZ$ such that for each
directed edge $u \rightarrow v$ of~$G$ we have $\mu(v) = \mu(u) - 1$.
  We call $G$ a \emph{graded DAG} if it has a level mapping.
\end{definition}

An example of graded DAG together with a level mapping is given in Figure~\ref{fig:u-all-DWT}.
It is easy to see (and shown in Proposition~1
of~\cite{odagiri2014greatest})
that a DAG $G$ is graded 
iff there are no two vertices $u$, $v$ and two directed paths $\chi$,
$\chi'$ in $G$ from $u$ to $v$ such that $\chi$ and $\chi'$ have
different lengths (in the terminology of~\cite{odagiri2014greatest},
$G$ does not have a \emph{jumping edge}).
Graded DAGs are related to the classical notion of graded ordered
set~\cite{schroder2016graded}, and the level mapping function 
has been called in the literature a \emph{depth
function}~\cite{odagiri2014greatest}, a \emph{grading
function}~\cite{schroder2016graded}, a \emph{set of
levels}~\cite{schroder2016graded}, or a \emph{rank
function}~\cite{stanley1997enumerative}.

To obtain such a level mapping, we can proceed by picking one vertex
in each connected component of~$G$,
mapping each of these vertices to level~0, and then
exploring $G$ by a breadth-first traversal and assigning the level of each
vertex according to the level of the vertex used to reach it, visiting
all edges and defining the image of each vertex. It is clear that this
process yields a level mapping of~$G$ unless it tries to assign two
different levels to the same vertex $v$, which cannot happen if there is no
jumping edge \cite[Proposition~1]{odagiri2014greatest}.

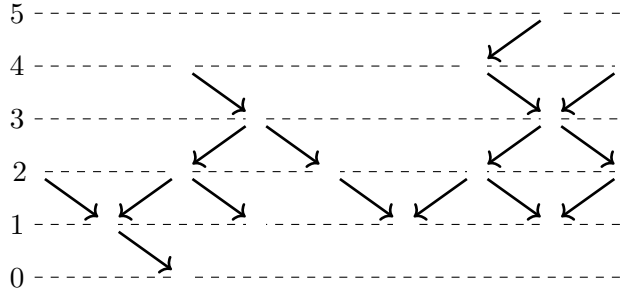
\begin{figure}
  \centering
  \begin{tikzpicture}[yscale=.7,xscale=.97]
        \node[fill = white, name = 02] at (0,2) {};
        \node[fill = white, name = 22] at (2,2) {};
	\node[fill = white, name = 33] at (3,3) {};
	\node[fill = white, name = 73] at (7,3) {};

	\node[fill = white, name = 2] at (-0.2,2) {$2$};

	\draw [dashed] (0,0) node[left] {$0$} -- (8,0);
	\draw [dashed] (0,1) node[left] {$1$} -- (8,1);
	\draw [dashed] (02)  -- (22);
	\draw [dashed] (22)  -- (8,2);
	\draw [dashed] (0,3) node[left] {$3$} -- (33);
	\draw [dashed] (33)  -- (73);
	\draw [dashed] (73) -- (8,3);

	\draw [dashed] (0,4) node[left] {$4$} -- (8,4);
	\draw [dashed] (0,5) node[left] {$5$} -- (8,5);

	\node[fill = white, name = 11] at (1,1) { };
	\node[fill = white, name = 20] at (2,0) { };
	\node[fill = white, name = 31] at (3,1) { };
	\node[fill = white, name = u1] at (2,4) { };
	\node[fill = white, name = u2] at (4,2) { };
	\node[fill = white, name = 51] at (5,1) { };
	\node[fill = white, name = 62] at (6,2) { };
	\node[fill = white, name = 71] at (7,1) { };
	\node[fill = white, name = v2] at (8,2) { };
	\node[fill = white, name = 64] at (6,4) { };
	\node[fill = white, name = v1] at (8,4) { };
	\node[fill = white, name = 75] at (7,5) { };

	\draw[line width = 1pt, ->]  (02) -- (11);
	\draw[line width = 1pt, ->]  (11) -- (20);
	\draw[line width = 1pt, ->]  (22) -- (11);
	\draw[line width = 1pt, ->]  (22) -- (31);
	\draw[line width = 1pt, ->]  (33) -- (22);
	\draw[line width = 1pt, ->]  (u1) -- (33);
	\draw[line width = 1pt, ->]  (33) -- (u2);
	\draw[line width = 1pt, ->]  (u2) -- (51);
	\draw[line width = 1pt, ->]  (62) -- (51);
	\draw[line width = 1pt, ->]  (62) -- (71);
	\draw[line width = 1pt, ->]  (v2) -- (71);
	\draw[line width = 1pt, ->]  (73) -- (62);
	\draw[line width = 1pt, ->]  (73) -- (v2);
	\draw[line width = 1pt, ->]  (64) -- (73);
	\draw[line width = 1pt, ->]  (v1) -- (73);
	\draw[line width = 1pt, ->]  (75) -- (64);
\end{tikzpicture}
  \caption{A DAG with 
  a level mapping (dashed lines), see
  Definition~\ref{def:levelmap}.}
\label{fig:u-all-DWT}
\end{figure}

  We will now use the notion of graded DAG to show:
\begin{propositionrep}
	\label{prp:u-all-DWT}
	\phomu{\All}{\UDWT} is PTIME.
\end{propositionrep}

\begin{proofsketch}
	We only give the idea when the query graph is connected and the graph
        instance $H$ is a \DWT{} (see Appendix for full proof).
        As we pointed out already, if the query graph $G$ is not a graded DAG,
        then it has a cycle or
        a pair of vertices joined by two directed paths of different lengths:
        then, from the structure of the \DWT instance graph, this clearly implies that $\Pr(G
        \leadsto H)=0$. So it suffices to study the case when $G$ is a graded
        DAG.

        As we explained earlier, we can then compute in PTIME a level mapping
        $\mu$ of~$G$.
        It is clear that, as $G$ is connected, the level
        mapping~$\mu$ is uniquely defined up to an additive constant. Hence, we shift $\mu$ so
        that the smallest value of its image is~$0$, and we then call the \emph{difference of
        levels} of~$G$ the largest value~$m$ in the image of~$\mu$.
      Note that
      $m$ is \emph{not} the maximal length
      from a root of $G$ to a leaf of $G$
      (see, e.g., Figure~\ref{fig:u-all-DWT}).
        We then claim that, on any possible world $H'$ of
        the \DWT{} instance graph
        $H$, the query graph $G$ is
        in fact equivalent to the \OWP query graph $\rightarrow^m$ of length~$m$. This allows us
        to conclude using Proposition~\ref{prp:u-DWT-PT}.

	One direction is easy to observe, because 
	$\mu$
        directly gives 
        a homomorphism from~$G$
        to~$\rightarrow^m$. For the converse, suppose that a homomorphism $h$ from $G$ to
	$H'$ exists. Because $G$ is connected and $H'$ is in \UDWT, the image of $h$ is actually a \DWT, call it $T$. Now it is easy to see that the image of a node that has level $m-i$ in
	$G$ has depth $i$ in $T$, so that $T$ (and so $H'$) contains the \OWP $\rightarrow^m$.
\end{proofsketch}

\begin{proof}
  Let $G$ be an arbitrary unlabeled graph and $(H,\pi)$ a probabilistic
  graph with $H\in \UDWT$. We
        observe that if $G$ contains a directed cycle, then it cannot have a
        homomorphism to a subgraph of~$H$ (which is necessarily acyclic), so
        $\Pr(G\hom H)=0$. Hence, it suffices to study the case
        where the query graph $G$ is a DAG.
        
        Likewise, if there are two vertices $u, v$ of~$G$ and directed paths
        $\chi, \chi'$ in $G$ from $u$ to~$v$ such that $\chi$ and~$\chi'$ have
        different lengths, then again $G$ cannot have a homomorphism to a
        subgraph of~$H$:
        indeed, any subgraph of $H$ is a directed forest and there is at most
        one directed path between each pair of nodes. So we can assume without
        loss of generality that there is no such pattern in~$G$, and $G$
        is therefore graded.

 Letting $\mu$ be a level mapping of $G$, we call the
 \emph{difference of levels} of $\mu$ the difference
 between the largest and smallest value of its image; the
 \emph{difference of levels} of~$G$ itself is the minimum
 difference of levels of a
 level mapping of~$G$. As the level mappings of~$G$ only differ in the 
 constant value that they add to all vertices of each connected
 component, the difference of levels can clearly be computed in PTIME by shifting each connected
 component so that its minimal level is zero, and computing the
 difference; we call the result of the shifting the \emph{minimal level
 mapping} of~$G$.

        \medskip

        Letting $m$ be the difference of levels of~$G$, we now make the
        following claim: \emph{in any subgraph $H'$ of~$H$, there is a
        homomorphism from $G$ to $H'$ if and only if $H'$ has a directed path of
        length $m$}.

	This claim implies the result. Indeed, we can first check in PTIME if
        $G$ has no cycles and has no pairs of paths of different lengths between
        two endpoints, and return 0 if the conditions are violated. We can then
        compute in PTIME the difference of levels $m$ of $G$ using the
        observations above. 
	Now, on any subgraph of~$H$, the query $G$ is equivalent to
        the \OWP{} 
        graph $\rightarrow^m$, so our result follows from
        Proposition~\ref{prp:u-DWT-PT}.

        \medskip

        All that remains is to prove the claim.
        We first note that it suffices
        to show the claim under the assumption that $G$ is connected. Indeed, if
        the claim is true for all connected $G$, then the claim is implied for
        arbitrary $G$ by considering each of its connected components, applying
        the claim, and observing that $G$ has a suitable homomorphism to $H'$
        iff each one of its connected components does, i.e., iff $H'$ has a
        directed path whose length is the maximal difference of levels of a
        connected component of~$G$, and this is precisely the difference of
        levels $m$ of~$G$.  Hence, we now prove the claim for connected $G$.

	We start with the backwards direction of the claim.
        It is easily seen that there is a
        homomorphism $h'$ from $G$ to the \OWP{} graph~$\rightarrow^m$. Indeed, we define $h'$
        according to the minimal level mapping $\mu$ of~$G$: we set $h'$ to map
        all the vertices whose level is $i$ to the $i$-th vertex
        of~$\rightarrow^m$. From the
        existence of~$h'$, we know that, whenever there is a homomorphism~$h$
        from $\rightarrow^m$ to~$H'$, then $h \circ h'$ is a homomorphism from
        $G$ to~$H'$, which shows the backwards implication.

	For the forward direction of the claim, suppose that there exists a homomorphism $h$ from $G$ to $H'$, and let
        $m$ be the difference of levels of $G$.
	Because $G$ is connected and $H'$ is in \UDWT, the image of $h$ is actually a \DWT, call it $T$. Now it is easy to see that the image of a node that has level $m-i$ in
	$G$ has depth $i$ in $T$, so that $T$ (and so $H'$) contains the \OWP $\rightarrow^m$.
	This finishes the proof of the converse and thus the proof of Proposition~\ref{prp:u-all-DWT}.
\end{proof}

\subsection{Disconnected Instances}
\label{sec:disconnected-instances}
We conclude our study of the disconnected case with 
the case of disconnected \emph{instance graphs}, which we show to be less
interesting than the disconnected \emph{query graphs} that we studied so far.
Specifically, when the query is \emph{connected}, \phom{}
on arbitrary instances can reduce in PTIME to
\phom{} of
the same queries on a corresponding class of connected instances:

\begin{lemma}
\label{lem:disconnected-instances}
        For any class of graphs $\calH$, let $\calH'$ be the class of 
        connected components of graphs in~$\calH$.
	Then for any class of \emph{connected} graphs $\calG$, \phoml{\calG}{\calH}  reduces in
        PTIME to \phoml{\calG}{\calH'},
        and \phomu{\calG}{\calH} reduces in PTIME to
        \phomu{\calG}{\calH'}.
\end{lemma}
\begin{inlineproof}
  Let $G \in \mathcal{G}$, $H \in \mathcal{H}$, and write $H = {H'_1 \sqcup \ldots
  \sqcup H'_n}$: we have 
  $H'_i \in \calH'$ for all $1 \leq i \leq n$.
  Let $\pi$ be a probability
 distribution over $H$:
  the independence assumption ensures that the
  edges of any $H'_i$ are pairwise independent from those of any~$H'_j$ for $i
  \neq j$.
 Now, as $G$ is connected, any image of a homomorphism from $G$ to~$H$ must actually
  be included in some~$H'_i$.
Thus, the computation of $\Pr(G\hom H)$ 
  reduces to that of the $\Pr(G\hom H'_i)$ for $1 \leq i \leq n$, as follows:
  \[\Pr(G\hom H)=1 - \prod_{1 \leq i \leq n} (1 -
\Pr(G\hom H'_i)). \qedhere \]
\end{inlineproof}

We last discuss the case when both the query and instance graphs are
disconnected.
Let us consider the results of Table~\ref{tab:unlabeled-conj_of} for connected
instance graphs. Clearly, any hardness results of a connected class
carries over to the corresponding disconnected class. Conversely, we have shown in
Proposition~\ref{prp:u-all-DWT} that \phomu{\All}{\UDWT} is PTIME; this implies
that
all tractable cases in Table~\ref{tab:unlabeled-conj_of} also hold
for
unions of the indicated instance classes, except \phomu{\UOWP}{\UPT} and \phomu{\UDWT}{\UPT}.
But we have noted
at the end of Section~\ref{sec:labeled-disconnected-queries}
that, in the unlabeled setting, \UOWP{} or \UDWT query graphs are equivalent to
\OWP{} query graphs: thus, Lemma~\ref{lem:disconnected-instances},
together with tractability of \phomu{\OWP}{\PT}, implies 
\phomu{\UOWP}{\UPT} and \phomu{\UDWT}{\UPT} are both in PTIME.
Hence, the results of
Table~\ref{tab:unlabeled-conj_of} also hold when instances are unions of the
indicated classes.

We have thus completed our study of 
$\phomL$ and $\phomU$ for
disconnected instances and/or disconnected
queries, We accordingly focus on connected queries and instances in the next two
sections.

\section{Labeled Connected Queries}\label{sec:l-connected-CQs}
In this section, we
focus on the \emph{labeled} setting, i.e., the $\phomL$ problem, for
classes of connected queries and instances.
Table~\ref{tab:labeled} shows the entire classification of the labeled setting 
for the classes that we consider.

\begin{table}[t]
  \caption{Tractability of \phomL{} in the connected case (Section~\ref{sec:l-connected-CQs})}

  \centering
  \begin{minipage}{.63\linewidth}
\begin{tabular}{c|ccccc}
    $\downarrow$$G$\qquad $H$$\rightarrow$ & \OWP & \TWP & \DWT & \PT & \Connected \\
    \hline
    \OWP & & & \vlineright{\ref{prp:l-1WP-DWT}} & \hardcell\ref{prp:l-1WP-PT} & \hardcell\\
    \hhline{~|~~|-|>{\hardline}--}\regline
    \TWP & & \vlineright{} & \hardcell\ref{prp:l-2WP-DWT} & \hardcell& \hardcell\\
    \DWT & & \vlineright{} & \hardcell\ref{prp:l-DWT-DWT} &  \hardcell& \hardcell\\
    \PT & & \vlineright{} & \hardcell & \hardcell & \hardcell\\
    \Connected & & \vlineright{\ref{prp:l-connected-2WP}} & \hardcell& \hardcell & \hardcell\\
    \hhline{~|--|>{\hardline}---}
\end{tabular}
\label{tab:labeled}

\tabexplanation
  \end{minipage}
\end{table}

Intuitively, we show intractability 
for polytree instance graphs, 
and for 
downward trees instance graphs when
the query graphs 
allow either two-wayness or branching.
Conversely, we
show tractability of one-way path query graphs on downward trees, and of
arbitrary connected queries on two-way path instances. We first present the
hardness results, and then the tractability results.

\subsection{Hardness Results}

We recall that, if we allow \emph{arbitrary}
connected unlabeled probabilistic instance graphs (or even just $4$-partite graphs), then computing the probability that there exists
a path of length $2$ is already \#P-hard: 
this is shown in~\cite{suciu2011probabilistic}, 
and we will state this result in our
context as Proposition~\ref{prp:unlabeled_1wp_on_connected} in the next section.
Hence, if we want to obtain PTIME complexity for \phom{}, we need to restrict the class of instances.
We can start by restricting the instances to be polytrees, but as we show, this
does not suffice to ensure tractability:

\begin{propositionrep}
	\label{prp:l-1WP-PT}
	\phoml{\OWP}{\PT} is \#P-hard.
\end{propositionrep}

To show this result, we will reduce from the problem of computing the
probability of a Boolean formula, which we now define:

\begin{definition}
        \label{def:booleanproba}
        Given a set of variables $\calX$ and a \emph{probability assignment} $\pi$
        mapping each variable $X$ in~$\calX$ to a rational probability $\pi(X) \in [0, 1]$, we define
        the \emph{probability} $\pi(\nu)$ of a valuation~$\nu: \calX \to \{0,1\}$ as 
        \[\pi(\nu) \colonequals \left(\prod_{X \in \calX,~\nu(X)=1}
        \pi(X)\right) \left(
        \prod_{X \in \calX,~\nu(X)=0} (1-\pi(X))\right).\]

        The \emph{Boolean probability computation problem} is defined as follows:
        given a Boolean formula $\phi$ on variables $\calX$ and a probability
        assignment $\pi$ on~$\calX$,
        compute the total probability of the valuations that satisfy~$\phi$, i.e., 
        $\Pr(\phi, \pi) = \sum_{\nu \text{~satisfies~} \phi} \pi(\nu)$.
\end{definition}

This problem is known to be \#P-hard, even under severe restrictions on the
formula $\phi$. We will use the \pptdnf{} formulation of the above problem, which
is \#P-hard~\cite{provan1983complexity,suciu2011probabilistic}:

\begin{definition}
        \label{def:pp2dnf}
        A \emph{positive DNF}
        is a Boolean formula~$\phi$ of the form
        \[\phi = \bigvee_{1 \leq i \leq m} \left( \bigwedge_{1 \leq j \leq n_i}
        X_{i,j} \right),\] i.e., it is
        a disjunction of
        (conjunctive) \emph{clauses} that are conjunctions of variables of~$\calX$.
        We assume that each variable of~$\calX$ occurs in~$\phi$, as we can
        eliminate the others without loss of generality.

        A \emph{positive partitioned 2-DNF} (PP2DNF) is intuitively a positive
        DNF $\phi$ on a
        partitioned set of variables where each clause contains one variable from
        each partition. Formally, the variables of~$\phi$ are $\calX \sqcup
        \calY$, where we write 
        $\mathcal{X} =
        \{X_1,\ldots,\allowbreak  X_{n_1}\}$ and $\mathcal{Y} = {\{Y_1,\ldots,
        \allowbreak Y_{n_2}\}}$,
        and 
        $\phi$ is of the form $\bigvee_{j=1 \ldots m} (X_{x_j} \land Y_{y_j})$
        with $1 \leq x_j \leq n_1$ and $1 \leq y_j \leq n_2$ for $1\leq j \leq
        m$.

        The \pptdnf{} problem is the Boolean probability computation problem
        when we impose
        that $\pi$ maps every variable to $1/2$, and
        that $\phi$ is a PP2DNF.
\end{definition}

We show Proposition~\ref{prp:l-1WP-PT}
by reducing from \pptdnf{}:

\begin{proofsketch}
  The full proof is in appendix; see Figure~\ref{fig:l-1WP-PT} for an
  illustration. From the PP2DNF formula $\phi$, we construct
  a $\PT$ probabilistic instance where each branch
  starting at the root describes a variable of the formula. The first edge is
  probabilistic and represents the choice of valuation. The edges are oriented
  upwards or downwards depending on whether the variable belongs to~$\calX$ or
  to~$\calY$. We add a special gadget at different depths of the branch
  to code the index of each of the clauses where the variable occurs.
  
  We code satisfaction of the formula by a query that tests for a path of a
  specific length that starts and ends with the gadget. The query has a match
  exactly on possible worlds where we have set two variables to true such that
  the sum of the depths of the gadgets corresponds to the query length: this
  happens iff the two variables occur in the same clause.
\end{proofsketch}

\begin{figure}
  \begin{tikzpicture}[xscale=1.52,yscale=1]
	\node[fill = white, name = H] at (-1.15,5) {$H$:};

	\node[fill = white, name = R] at (3.5,5) {$R$};

	\node[fill = white, name = X1] at (0.5,3) {$X_1$};
	\node[fill = white, name = X2] at (3,3) {$X_2$};

	\node[fill = white, name = Y1] at (4,3) {$Y_1$};
	\node[fill = white, name = Y2] at (6.5,3) {$Y_2$};

	\node[fill = white, name = X13] at (0.5,1.5) {$X_{1,3}$};
	\node[fill = white, name = X12] at (0.5,0) {$X_{1,2}$};
	\node[fill = white, name = X11] at (0.5,-1.5) {$X_{1,1}$};

	\node[fill = white, name = X23] at (3,1.5) {$X_{2,3}$};
	\node[fill = white, name = X22] at (3,0) {$X_{2,2}$};
	\node[fill = white, name = X21] at (3,-1.5) {$X_{2,1}$};

	\node[fill = white, name = Y11] at (4,1.5) {$Y_{1,1}$};
	\node[fill = white, name = Y12] at (4,0) {$Y_{1,2}$};
	\node[fill = white, name = Y13] at (4,-1.5) {$Y_{1,3}$};

	\node[fill = white, name = Y21] at (6.5,1.5) {$Y_{2,1}$};
	\node[fill = white, name = Y22] at (6.5,0) {$Y_{2,2}$};
	\node[fill = white, name = Y23] at (6.5,-1.5) {$Y_{2,3}$};

	\node[fill = white, name = A11] at (-1,-2.5) {$A_{1,1}$};
	\node[fill = white, name = A12] at (-1,-1) {$A_{1,2}$};

	\node[fill = white, name = A23] at (1.5,0.5) {$A_{2,3}$};

	\node[fill = white, name = B12] at (5.5,-1) {$B_{1,2}$};

	\node[fill = white, name = B21] at (8,0.5) {$B_{2,1}$};
	\node[fill = white, name = B23] at (8,-2.5) {$B_{2,3}$};

	\draw[dashed, line width = 1pt, ->]  (X1) -- (R) node[midway, left=2pt, above = 2pt, fill=white] {$S$};
	\draw[dashed, line width = 1pt, ->]  (X2) -- (R) node[midway, left=2pt, fill=white] {$S$};

	\draw[dashed, line width = 1pt, ->]  (R) -- (Y1) node[midway, right=2pt, fill=white] {$S$};

	\draw[dashed, line width = 1pt, ->]  (R) -- (Y2) node[midway, right=2pt, above = 2pt, fill=white] {$S$};

	\draw[line width = 1pt, <-]  (X1) -- (X13) node[midway, left=1pt, fill=white] {$S$};
	\draw[line width = 1pt, <-]  (X13) -- (X12) node[midway, left=1pt, fill=white] {$S$};
	\draw[line width = 1pt, <-]  (X12) -- (X11) node[midway, left=1pt, fill=white] {$S$};

	\draw[line width = 1pt, <-]  (X2) -- (X23) node[midway, left=1pt, fill=white] {$S$};
	\draw[line width = 1pt, <-]  (X23) -- (X22) node[midway,left=1pt, fill=white] {$S$};
	\draw[line width = 1pt, <-]  (X22) -- (X21) node[midway,left=1pt, fill=white] {$S$};

	\draw[line width = 1pt, ->]  (Y1) -- (Y11) node[midway, right=1pt, fill=white] {$S$};
	\draw[line width = 1pt, ->]  (Y11) -- (Y12) node[midway, right=1pt, fill=white] {$S$};
	\draw[line width = 1pt, ->]  (Y12) -- (Y13) node[midway, right=1pt, fill=white] {$S$};

	\draw[line width = 1pt, ->]  (Y2) -- (Y21) node[midway, right=1pt, fill=white] {$S$};
	\draw[line width = 1pt, ->]  (Y21) -- (Y22) node[midway, right=1pt, fill=white] {$S$};
	\draw[line width = 1pt, ->]  (Y22) -- (Y23) node[midway, right=1pt, fill=white] {$S$};

	\draw[line width = 1pt, ->]  (A11) -- (X11) node[midway, left=6pt, above = 1pt, fill=white] {$T$};
	\draw[line width = 1pt, ->]  (A12) -- (X12) node[midway, left=6pt, above = 1pt, fill=white] {$T$};

	\draw[line width = 1pt, ->]  (A23) -- (X23) node[midway, left=6pt, above = 1pt, fill=white] {$T$};

	\draw[line width = 1pt, ->]  (Y12) -- (B12) node[midway, right=6pt, above = 1pt, fill=white] {$T$};

	\draw[line width = 1pt, ->]  (Y21) -- (B21) node[midway, right=6pt, above = 1pt, fill=white] {$T$};
	\draw[line width = 1pt, ->]  (Y23) -- (B23) node[midway, right=6pt, above = 1pt, fill=white] {$T$};

        \node[fill = white, name = G] at (3.5,-3.4) {$G$:
        $\xrightarrow{\,~T\,~}\xrightarrow{\,~S\,~}\xrightarrow{\,~S\,~}\xrightarrow{\,~S\,~}\xrightarrow{\,~S\,~}\xrightarrow{\,~S\,~}\xrightarrow{\,~S\,~}\xrightarrow{\,~T\,~} $};

\end{tikzpicture}
\centering
\caption{Illustration of the proof of Proposition~\ref{def:pp2dnf} for the
PP2DNF formula $X_1 Y_2 \lor X_1 Y_1 \lor X_2 Y_2$. Dashed edges have probability $\frac{1}{2}$, all others have probability $1$.}
\label{fig:l-1WP-PT}
\end{figure}
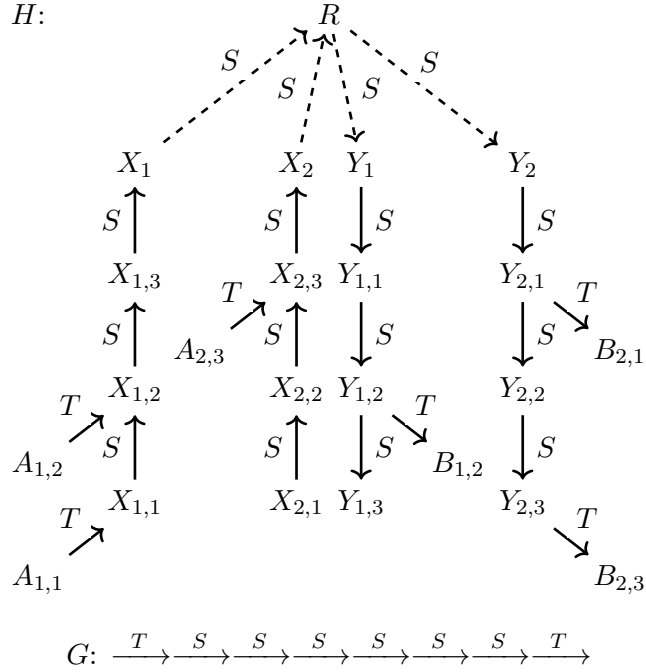

\begin{proof}
	We reduce from the \#P-hard problem \pptdnf.
	From the formula $\phi$,
        we construct the following $\{S, T\}$-labeled probabilistic graph $H$
        (an example of this construction is presented in
        Figure~\ref{fig:l-1WP-PT}):

	\begin{itemize}
		\item The vertices of $H$ are $\{R\} \sqcup \{X_1,\ldots
                  X_{n_1}\} \allowbreak \sqcup \{Y_1,\ldots,\allowbreak Y_{n_2}\}
                  \allowbreak \sqcup \{X_{i,j} \mid 1 \leq i \leq n_1,~1\leq j \leq m\}
			\sqcup \{Y_{i,j} \mid 1 \leq i \leq n_2,~ 1\leq j
                        \leq m\} \sqcup \{A_{x_j,j} \mid 1 \leq j \leq
                        m\} \sqcup \{B_{y_j,j} \mid 1 \leq j \leq m\}$.
		\item The edges of $H$, all of which have probability $1$ except when specified, are: 
			\begin{itemize}
				\item $X_i \xrightarrow{S} R$ for all $1 \leq i
                                  \leq n_1$ and $R \xrightarrow{S} Y_i$ for all $1 \leq i \leq n_2$, all having probability $\frac{1}{2}$ and intuitively coding the valuation of each variable; 
				\item For all $1 \leq i \leq n_1$, the edge
                                  $X_{i,m} \xrightarrow{S} X_i$ and the edges
                                  $X_{i,j} \xrightarrow{S} X_{i,j+1}$ for all $1 \leq j \leq m-1$;
				\item For all $1 \leq i \leq n_2$, the edge $Y_i
                                  \xrightarrow{S} Y_{i,1}$ and the edges
                                  $Y_{i,j} \xrightarrow{S} Y_{i,j+1}$ for all $1 \leq j \leq m-1$;
                                \item For all $1 \leq j \leq m$, the
                                  edges $A_{x_j,j} \xrightarrow{T}
                                  X_{x_j,j}$ and $Y_{y_j,j}
                                  \xrightarrow{T} B_{y_j,j}$, intuitively
                                  indicating that variables $X_{x_j}$ and
                                  $Y_{y_j}$ belong to clause $j$.
			\end{itemize}	
	\end{itemize}
	The $\{S, T\}$-labeled graph $G$ is then $\xrightarrow{T} (\xrightarrow{S})^{m+3} \xrightarrow{T}$.
	It is clear that $G$ is a \OWP{} query graph, $H$ is a polytree and that both can be constructed in PTIME from $\phi$.
        We now show that $\Pr(G\hom H)$ is exactly the number of satisfying assignments of~$\phi$ divided by $2^n$, so that the computation of one reduces in PTIME to the 
	computation of the other, concluding the proof. 
	To see why, we define a bijection between the valuations $\nu$ of $\{X_1,\ldots,X_{n_1}\} \sqcup \{Y_1,\ldots,Y_{n_2}\}$ to the possible worlds $H'$ 
	of $H$ that have non-zero probability, in the expected way: keep 
	the edge $X_i \xrightarrow{S} R$ (resp., $R \xrightarrow{S} Y_i$)
        iff $X_i$ (resp., $Y_i$) is assigned to true in the valuation. We then show that there is a homomorphism
	from $G$ to $H'$ if and only if $\phi$ evaluates to true under $\nu$.
        
        Indeed, if there is a homomorphism from $G$ to $H'$, then by considering the only possible matches of the $T$-edges, one can check easily that
	the image of the match in $H'$ must be of the following form for some 
	$1 \leq j \leq m$: $A_{x_j,j} \xrightarrow{T} X_{x_j,j} \xrightarrow{S} X_{x_j,j+1} \xrightarrow{S} \cdots \xrightarrow{S} X_{x_j,m} \xrightarrow{S} X_{x_j} \xrightarrow{S} R \xrightarrow{S}
        Y_{y_{j'}} \xrightarrow{S} Y_{y_{j'}, 1} \xrightarrow{S} Y_{y_{j'}, 2} \xrightarrow{S} \cdots \xrightarrow{S} Y_{y_{j'}, {j'}} \xrightarrow{T} B_{y_{j'}, {j'}}$; further, from the length of the $S$-path we must have $(m-j)+4+(j'-1) = m+3$, so that we must have $j = j'$. Then, by construction,
	$X_{x_j}$ and~$Y_{y_j}$ belong to clause $j$, so the valuation satisfies~$\phi$.
	Conversely, suppose that the valuation satisfies $\phi$, then for some $1 \leq j \leq m$ we know that $X_{x_j}$ and $Y_{y_j}$ are assigned to true by the valuation, 
	and so we can build the homomorphism as above from $G$ to $H'$.
\end{proof}

Hence, restricting instances to polytrees is not sufficient to ensure
tractability, even for $\OWP$ query graphs. We must thus restrict the instance
further, by disallowing one of the two remaining features, namely branching and
two-wayness.
The first option of disallowing branching, i.e., requiring the
instance to be a $\TWP$, is studied in Section~\ref{sec:tractability} below,
where we show that the problem is tractable for arbitrary
query graphs.

The second option is to forbid two-wayness on the instance,
i.e., restrict it to be a $\DWT$. In this case, we first show that
intractability holds even when we also forbid two-wayness in the query graph,
i.e., we also restrict it to be a $\DWT$.
The result follows from our earlier work on
the combined complexity of query
evaluation~\cite{amarilli2017combined,amarilli2017combinedb}:

\begin{propositionrep}[\cite{amarilli2017combinedb}]
\label{prp:l-DWT-DWT}
\phoml{\DWT}{\DWT} is \#P-hard.
\end{propositionrep}
\begin{proof}
  The proof is almost the same as that of Proposition~36
  of~\cite{amarilli2017combinedb}, straightforwardly adapted to our setting of probabilistic
  graphs (in particular replacing the unary relation~$R$ by a binary relation),
  by observing that the probabilistic instance defined in this proof is actually
  a $\DWT$ (beyond having treewidth~1), and that the query actually corresponds
  to a $\DWT$ graph (beyond being $\alpha$-acyclic).
\end{proof}

If we forbid branching in the query graph instead of two-wayness, requiring it
to be a $\TWP$, then intractability still holds, which also follows from our
earlier results:

\begin{propositionrep}[\cite{amarilli2017combinedb}]
\label{prp:l-2WP-DWT}
\phoml{\TWP}{\DWT} is \#P-hard.
\end{propositionrep}
\begin{proof}
  The proof is almost the same as that of Proposition~38
  of~\cite{amarilli2017combinedb}, again adapted to our setting of probabilistic
  graphs, with one small modification: we do not 
  materialize edges $b \xrightarrow{S_-} a$ in the instance graph for each edge
  $a \xrightarrow{S} b$ in the instance, and
  instead modify the query to replace all edges $x \xrightarrow{S_-} y$ by edges $x
  \xleftarrow{S} y$. This ensures that the query is a $\TWP$ and the instance is
  a $\DWT$, and hardness is shown similarly to the original proof.
\end{proof}

Thus, on $\DWT$ instances, the only remaining case is when the query is a
one-way path. We will now show in the section below that this case is tractable, in
addition to the case of arbitrary queries on $\TWP$ instances that we left open
above.

\subsection{Tractability Results}
\label{sec:tractability}
The general proof technique to obtain PTIME combined complexity in this section
is inspired by the probabilistic database literature~\cite{suciu2011probabilistic}:
compute the \emph{lineage} of~$G$ on $H$ as a Boolean formula in positive disjunctive
normal form (DNF), then compute its probability. Let us first define
\emph{lineages}:

\begin{definition}
  \label{def:lineage}
   Let $G$ be a query graph and $(H, \pi)$ be a probabilistic graph with edge set~$E$.
   For any valuation $\nu: E \to \{0, 1\}$, we denote by~$\nu(H)$ the possible
   world of~$H$ where each edge $e \in E$ is kept iff $\nu(e) = 1$.
   Letting $\phi$ be a
   Boolean function whose variables are the edges of~$E$, we say that $\phi$
   \emph{captures the lineage of $G$ on $H$} if, for any valuation $\nu: E \to
   \{0, 1\}$, the function $\phi$ evaluates to~$1$ under~$\nu$ iff we have $G
   \leadsto \nu(H)$.
\end{definition}

Lineage representations allow us to reduce the
$\phom$ problem to the Boolean probability computation problem on the lineage
function (recall Definition~\ref{def:booleanproba}).
Formally, for any query graph $G$ and
probabilistic graph $(H, \pi)$, given a Boolean function $\phi$ that captures the lineage of$~G$
on~$H$, we compute the answer to $\phom$ on~$G$ and $(H,\pi)$ as the probability $\Pr(\phi,
\pi)$ of $\phi$ under~$\pi$: it is immediate by definition that these two
quantities are equal.

Of course, computing a lineage representation does not generally suffice to show
tractability, because, as we explained earlier, the Boolean probability
computation problem is generally intractable. However, computing a Boolean
lineage allows us to leverage the known tractable classes of Boolean
formulas.
Specifically, we will show how to use the class of \emph{$\beta$-acyclic
  positive DNF
formulas}, which are known to be tractable~\cite{brault2015understanding}.
We define this notion, by first recalling
the notion of a \emph{$\beta$-acyclic hypergraph}, and then defining a
\emph{$\beta$-acyclic positive DNF}:

\begin{definition}
        \label{def:beta-acyclic}
	A \emph{hypergraph} $\mathcal{H} = (V,E)$ is a finite set~$V$ of vertices and a set~$E$ of non-empty subsets of~$V$, called \emph{hyperedges}.
	For $v \in V$, we write $\mathcal{H} \setminus v$ for the hypergraph $(V \setminus \{v\}, E')$ where $E'$ is
	$\{e \setminus \{v\} \mid e \in E\} \setminus \{ \emptyset \}$.

        A vertex $v \in V$ of~$\mathcal{H}$ is called a \emph{$\beta$-leaf}
        \cite{brault2014hypergraph} if the set of hyperedges that contain it,
        i.e., $\{e \in E \mid v \in e\}$, is totally
	ordered by inclusion. In other words, we can write $\{e \in E
          \mid v \in
        e\}$ as $(e_1, \ldots, e_k)$ in a way that ensures that $e_i \subseteq
        e_{i+1}$ for all $1 \leq i < k$.
        
	A $\beta$-elimination order for a hypergraph
        $\mathcal{H} = (V,E)$
        is defined inductively as follows:
        \smallskip
	\begin{itemize}
		\item if $E = \emptyset$, then the empty tuple is a $\beta$-elimination order for $\mathcal{H}$;	
		\item otherwise, a tuple $(v_1, \ldots, v_n)$ of vertices of $\mathcal{H}$ is a $\beta$-elimination order for $\mathcal{H}$
                  if $v_1$ is a $\beta$-leaf in~$\mathcal{H}$ and $(v_2, \ldots, v_n)$ is a $\beta$-elimination order for $\mathcal{H} \setminus v_1$.
	\end{itemize}
        \smallskip
        The hypergraph $\mathcal{H}$ is \emph{$\beta$-acyclic} if there is a $\beta$-elimination order for $\mathcal{H}$.
\end{definition}

We can see a positive DNF (recall Definition~\ref{def:pp2dnf}) as a hypergraph
of clauses on the variables, and introduce the notion of
\emph{$\beta$-acyclic positive DNFs} accordingly:

\begin{definition}
        The \emph{hypergraph} $\mathcal{H}(\phi)$ of a positive DNF 
	on variables~$\calX$
        has $\calX$ as vertex set and has one hyperedge per clause, i.e., we have
        $\mathcal{H}(\phi) \colonequals (\calX,
        E)$ with $E \colonequals \{ \{X_{i,j} \mid 1 \leq j \leq n_i\} \mid 1 \leq i \leq m\}$. 
        We say that the positive DNF $\phi$ is \emph{$\beta$-acyclic} if $\mathcal{H}(\phi)$ is $\beta$-acyclic.
\end{definition}

It follows directly from results by Brault-Baron, Capelli, and Mengel about the
$\beta$-acyclic \cspd{} problem~\cite{brault2015understanding}
that we can tractably compute the probability of $\beta$-acyclic positive
DNFs:

\begin{theoremrep}
	\label{thm:prob-acyclic-DNF}
	The Boolean probability computation problem is in PTIME
        when restricted to $\beta$-acyclic positive DNF formulas.
\end{theoremrep}

\begin{proofsketch}
        The \cspd{} problem studied in \cite{brault2015understanding} is about
        computing a partition function over the hypergraph, under weighted
        constraints on hyperedges: it generalizes the problem of counting the
        number of valuations of $\beta$-acyclic formulae in conjunctive normal
        form (CNF) by~\cite[Lemma~3]{brault2015understanding}.
        We show how the result extends
        to $\beta$-acyclic positive DNF, using de Morgan's law, and to probability computation
        for weighted variables, using additional constraints on singleton
        variable sets.
\end{proofsketch}

\begin{proof}
  We reduce our Boolean probability computation problem to the problem of
  $\beta$-acyclic \cspd{} of~\cite{brault2015understanding}, which they show to
  be in PTIME (Theorem~26 of~\cite{brault2015understanding}). We will explain
  how probability computation in the sense of Definition~\ref{def:booleanproba}
  can be encoded in their setting, by a variant of their own encoding
  (in Lemma~3 of~\cite{brault2015understanding}): we give a full proof for
  completeness.

	First, we recall their definition of \cspd{} (Definitions~1 and~2 in~\cite{brault2015understanding}) in the case of a Boolean domain.
        We denote by $\QQp$ the nonnegative rational numbers.
        Denote by $\{0, 1\}^\calX$ the set of functions from $\calX$
        to~$\{0,1\}$, i.e., the Boolean valuations of~$\calX$.
        For $\nu \in \{0, 1\}^\calX$ and $\calY \subseteq \calX$, we denote by
        $\restr{\nu}{\calY}$ the restriction of~$\nu$ to~$\calY$.
        A \emph{weighted constraint (with default value)} on variables~$\calX$
        is a pair
        $c=(f,\mu)$ that consists of a
        function $f: S \to \QQp$ for some subset~$S$ of~$\{0, 1\}^\calX$,
        called the \emph{support} of $c$, 
        and a \emph{default value} $\mu \in \QQp$;
        we write $\mathrm{var}(c) \colonequals \calX$.
        The constraint $c$ induces a total function on $\{0, 1\}^\calX$, also
        denoted $c$, that maps $\nu \in \{0,1\}^\calX$ to~$f(\nu)$ if $\nu \in S$, and to~$\mu$ otherwise.
        The \emph{size} of $c$ is $|c| = |S| \times |\calX|$. Intuitively, a
        constraint with default value
        assigns a weight in~$\QQp$ to all valuations of~$\calX$,
        but the default value mechanism allows us to avoid writing explicitly the complete table of this mapping.
	
	An instance of the \cspd{} problem then consists of a finite set $I$ of weighted constraints.
	The size of $I$ is $|I| \colonequals \sum_{c \in I} |c|$, and we write
        $\mathrm{var}(I) \colonequals \bigcup_{c \in I} \mathrm{var}(c)$.
	The output of the problem is the \emph{partition function}
        \[w(I) = \sum_{\nu \in \{0, 1\}^{\mathrm{var}(I)}} \prod_{c \in I}
        c(\restr{\nu}{\mathrm{var}(c)}).\]

        The \emph{hypergraph} $\mathcal{H}(I)$ of the \cspd{} instance $I$ (defined in Section~2.2 of~\cite{brault2015understanding}) is the hypergraph $(\mathrm{var}(I), E_I)$ where
        $E_I = \{\mathrm{var}(c) \mid c \in I \}$. We say that $I$ is \emph{$\beta$-acyclic} if $\mathcal{H}(I)$ is a $\beta$-acyclic hypergraph (recall Definition~\ref{def:beta-acyclic}), and we call
        \emph{$\beta$-acyclic \cspd{}} the problem \cspd{} restricted to $\beta$-acyclic instances.
	By Theorem~26 of~\cite{brault2015understanding}, the problem $\beta$-acyclic \cspd{} is in PTIME.

        We now explain how to reduce the probability computation problem to the $\beta$-acyclic \cspd{} problem.
	Let $\phi = \bigvee_{1 \leq i \leq m} \left( \bigwedge_{1 \leq j \leq
        n_i} X_{i,j} \right)$ be a Boolean 
	$\beta$-acyclic DNF on variables~$\calX$, with probabilities $\pi(X) \in
        [0,1]$ for each $X \in \calX$.
        We construct in linear time from~$\phi$ and~$\pi$ the variable set
        $\calX' \colonequals \{X' \mid X \in \calX\}$, the CNF formula $\phi' \colonequals \bigwedge_{1 \leq i \leq m} \left( \bigvee_{1 \leq j \leq n_i} X'_{i,j} \right)$, and the probability valuation
        $\pi'$ on~$\calX'$ defined by $\pi'(X'_{i,j}) = 1-\pi(X_{i,j})$. By De Morgan's duality law, $\phi'$ is equivalent to the negation of~$\phi$, so that we have
	 $\Pr(\phi, \pi) = 1-\Pr(\phi', \pi')$; hence, the probability
         computation problem for~$\phi$ and~$\pi$ reduces in PTIME to the
         same problem for $\phi'$ and~$\pi'$.

	We then construct in linear time a $\beta$-acyclic \cspd{} instance $I$ such that $\Pr(\phi', \pi') = w(I)$, which concludes the proof.
        For each variable $X' \in \calX'$, we define a weighted constraint
        $c_{X'}$ on variables $\{X'\}$
        by $c_{X'}(X' \mapsto 1) = \pi'(X')$ and $c_{X'}(X' \mapsto 0) =
        1-\pi'(X')$, which codes the probability of the variables.
        Now, for each clause $1 \leq i \leq m$, just like in
        Lemma~3 of~\cite{brault2015understanding}, we define a weighted
        constraint $c_i=(f_i,1)$ with default value~1 whose variables are
        $\{X'_{i,j} \mid 1 \leq j \leq n_i\}$, i.e., those that occur in the clause:
        $f_i(\nu)$ is $0$ for the (unique) valuation that sets all variables of the clause to~$0$, 
        intuitively coding the constraint of the clause.
        From the fact that $\phi$ was $\beta$-acyclic, it is clear that $I$ is also $\beta$-acyclic.
        Now, the result $w(I)$ of the partition function sums over all
        valuations of the variables of~$I$, namely the variables~$\calX'$
        of~$\phi'$. Whenever a valuation does not satisfy some clause $1 \leq i
        \leq m$, the weighted constraint $c_i$ will give it weight~0, hence
        ensuring that the product evaluates to~0, so we can restrict the sum to
        valuations that satisfy~$\phi'$: such valuations are given weight~$1$ by all weighted
        constraints $c_i$. Now, it is easy to see that the weight of valuations
        $\nu$ that satisfy $\phi$ is their probability $\pi'(\nu)$, as each
        $c_{X'}$ gives them weight $\pi'(X')$ or $1-(\pi'(X'))$ depending on
        whether $\nu(X')$ is~1 or~0.
        Hence, we have reduced the probability computation problem for $\beta$-acyclic DNF formulas
        to $\beta$-acyclic \cspd{} in PTIME, which concludes the proof.
\end{proof}

We will then use the tractability of $\beta$-acyclic formulas to show PTIME combined
complexity results for our $\phomL$ problem. The first result that we show
is tractability for labeled \OWP{} query graphs on \DWT{} probabilistic instance
graphs:\footnote{The connection to $\beta$-acyclicity in this context is due to Florent
Capelli.}

\begin{propositionrep}
	\label{prp:l-1WP-DWT}
	\phoml{\OWP}{\DWT} is PTIME.
\end{propositionrep}

\begin{proofsketch}
Intuitively, the proof proceeds in three steps. The first step is to enumerate
all candidate minimal matches of the query graph in the instance graph, i.e.,
subgraphs of the instance graph to which the query graph could have a homomorphism, and
which are minimal for inclusion. As the query
graph is a path, we know that the minimal matches are downward paths in the $\DWT$
instance: hence, as each vertex of the $\DWT$ instance can be the lowest vertex
of at most one match, there are polynomially many matches to consider.
  
  The
second step is to decide which ones of these matches are actually a match of the
  query, by considering the labels: as both the query graph and the match are a $\OWP$,
this is straightforward. These first two steps produce a positive DNF that captures the lineage
of the query graph on the instance in the standard sense.
  
  The third step is to notice that this
lineage expression is $\beta$-acyclic: this is because its variables can be eliminated
by considering the nodes of the instance $\DWT$ in a bottom-up fashion.
\end{proofsketch}

\begin{proof}
	Let $G \defeq u_1 \xrightarrow{R_1} \cdots \xrightarrow{R_{m-1}} u_m$ be the \OWP{} query (where all $R_i$ are not necessarily distinct), and $H$ be the downwards tree instance.
	The idea is to construct the lineage of~$G$ on~$H$ as a $\beta$-acyclic
        DNF $\phi$, so that we can conclude with Theorem~\ref{thm:prob-acyclic-DNF}.
	It is clear that any match of $G$ can only be a downwards path of $H$, hence we construct~$\phi$ as follows:
	for every downwards path $a_1 \xrightarrow{R'_1} \cdots \xrightarrow{R'_{m-1}} a_m$ of length $m$ of $H$ (their number is linear in $|H|$ because each path is uniquely defined by the choice of~$a_m$) check if the path is a match
	of $G$ (i.e, check that $R_i = R'_i$ for $1 \leq i \leq m-1$), and if it is the case then create a new clause of $\phi$ whose variables are all
	the facts $a_i \xrightarrow{R_i} a_{i+1}$ for $1 \leq i \leq m-1$.

        The formula $\phi$ thus obtained is then a DNF representation of the lineage of $H$ on~$G$, and has been built in time $O(\card{H} \cdot \card{G})$, i.e., in PTIME. 
	We now justify that $\phi$ is $\beta$-acyclic by giving a
        $\beta$-elimination order for~$\phi$: while $H$ still has edges,
        repeatedly pick a leaf $b$ of~$H$ and, letting $a$ be the parent of~$b$, eliminate
	the variable $a \xrightarrow{R} b$ from $\phi$. 
	Such a variable will always be a $\beta$-leaf, as any set of downwards paths of a downwards tree all ending at
	a leaf is necessarily ordered by inclusion. From the above, the fact
        that $\phi$ is $\beta$-acyclic suffices to conclude the proof.
\end{proof}

Interestingly, we were not able to prove this result using 
tree automata-based dynamic programming approach (like we will do later for
Proposition~\ref{prp:u-1WP-PT}).

The second result that we show is tractability when restricting the instance to be a \TWP{},
and allowing arbitrary \emph{connected} queries (remember from
Proposition~\ref{prp:l-conj_of_1WP-1WP} that the problem is hard even on $\OWP$
instances if we allow \emph{disconnected} queries):

\begin{propositionrep}
	\label{prp:l-connected-2WP}
	\phoml{\Connected}{\TWP} is PTIME.
\end{propositionrep}
\begin{proof}
  First of all, notice that, as the query graph $G$ is connected, the image of a homomorphism from the query $G$ to the \TWP instance $H$ is necessarily a connected component of $H$.
  Moreover, each connected component of $H$ is also a \TWP and there are $O(|H|^2)$ of them.
  We then proceed as follows.
  For every connected subpath ${C = a_1 - \cdots - a_n}$ (with each~$-$
  being either $\xrightarrow{R}$ or $\xleftarrow{R}$ for some binary relation $R$ in $H$) of~$H$,
  we check if there is a homomorphism from $G$ to $C$.
  This can be done in PTIME by Theorem~\ref{thm:xbar},
  because $C$ trivially has the
  $\underline{X}$-property w.r.t.\ the 
  total order $a_1 < a_2 < \cdots < a_n$: the only possibility for
  $(n_0,n_3)$ and $(n_1,n_2)$ to be edges of $C$ when $n_0$ comes before~$n_1$
  and $n_2$~comes before~$n_3$ is if $n_0=n_2$ and $n_1=n_3$, in
  which case it cannot hold that $n_0\xrightarrow{R}n_3$ and
    $n_1\xrightarrow{R}n_2$ at the same time, because we disallow multi-edges.
  If there is such a homomorphism, then we create a new clause of $\phi$ whose variables are all
  the facts that belong to~$C$.

  From this, we obtain in PTIME a positive DNF $\phi$ that captures the lineage of $G$ on $H$. 
  We now justify that $\phi$ is $\beta$-acyclic by giving a
  $\beta$-elimination order for $\phi$, by an argument similar to the
  proof of Proposition~\ref{prp:l-1WP-DWT}: repeatedly eliminate
  a variable $a - b$ from $\phi$ and this fact from $H$, where $b$ is an endpoint of $H$. 
  Indeed such a variable will always be a $\beta$-leaf, as any set of
  connected component of $H$ including $a - b$ is necessarily ordered by
  inclusion. Hence, $\phi$ is $\beta$-acyclic, which allows us to
  conclude.
\end{proof}

To show this result, we follow the same scheme as in the proof of
Proposition~\ref{prp:l-1WP-DWT} above:
\begin{inparaenum}[(i)]
\item enumerate all candidate matches;
\item check whether they are indeed matches; and
\item 
argue that the resulting lineage is $\beta$-acyclic.
\end{inparaenum} For the first step, there are polynomially many candidate
matches to consider, because matches are necessarily connected subgraphs
of the instance graph~$H$, that are uniquely defined by their endpoints: this is
where we use connectedness of the query. For the third step, the
resulting lineage is $\beta$-acyclic for the same reason as in
Proposition~\ref{prp:l-1WP-DWT}, as we can eliminate variables following the
order of the path~$H$: all connected subpaths containing an endpoint of the path are
ordered by inclusion. What changes, however, is the second step: from the
quadratically many possible matches, to
compute the lineage expression, we must decide which ones are actually matches.

Deciding this for each subpath amounts to testing, given the connected query
graph~$G$ and a candidate match~$H'$, whether $G \leadsto H'$, in the
non-probabilistic sense. This graph homomorphism problem is generally
intractable, but here the minimal match $H'$ is a $\TWP$ (as it is a subpath of~$H$),
so it turns out to enjoy combined
tractability. The corresponding result was first shown 
by Gutjahr~\cite{gutjahr1992polynomial} 
for \emph{unlabeled} graphs, when the instance graph is a path, or for more
general instances satisfying a condition called the $\underline{X}$-property;
this was generalized to labeled graphs by Gottlob, Koch, and Schulz
in~\cite{gottlob2006conjunctive}. We recall here the definition of this property:

\begin{definition}[(Definition~3.2 of \cite{gottlob2006conjunctive})]
\label{def:X-property}
  Let $H=(V,E,\lambda)$ be a directed graph with labels on $\sigma$, let
  $R\in\sigma$,
  and let $<$ be a total order on~$V$. 
  Then $R$ is said to have the \emph{$\underline{X}$-property} w.r.t.\ $<$ if
  for all $n_0,n_1,n_2,n_3 \in V$ such that $n_0 < n_1$
  and $n_2 < n_3$, if we have $n_0\xrightarrow{R}n_3$ and
  $n_1\xrightarrow{R}n_2$ then we also have
  $n_0\xrightarrow{R}n_2$.
  $H$ is said to have the \emph{$\underline{X}$-property} w.r.t.\ $<$ if it is the case
  of each label~$R$.
\end{definition}

\begin{theorem}(Theorem~3.5 of \cite{gottlob2006conjunctive}, extending
  Theorem~3.1 of \cite{gutjahr1992polynomial})
  \label{thm:xbar}
  Given a labeled query graph $G$, and given a labeled directed graph $H$ with the
  $\underline{X}$-property w.r.t.\ some order~$<$,
  we can determine in time $O(|H| \times |G|)$ whether $G \leadsto H$.
\end{theorem}

We can use this result to check, for all connected subpaths of the $\TWP$ instance
graph, whether the query graph has a homomorphism to the subpath.
This leads to the following sketch for the proof of
Proposition~\ref{prp:l-connected-2WP} (the full proof is in Appendix):

\begin{proofsketch}
  We proceed following the three-step process outlined above.
  We first enumerate the possible query matches in the instance, i.e., the
  quadratic number of connected subpaths. Second, we test for each subpath 
  $a_i - \cdots - a_{i+k}$ whether it satisfies the query. We can do so
  tractably because the subpath clearly has the $\underline{X}$-property w.r.t.\
  the order $a_i < \cdots < a_{i+k}$: using the notation of
  Definition~\ref{def:X-property}, there are in fact no $n_0,n_1,n_2,n_3$ that satisfy the
  conditions. Third, having computed the resulting DNF, we compute its
  probability using $\beta$-acyclicity, eliminating variables in the order of
  the path as we explained above.
\end{proofsketch}

\section{Unlabeled Connected Queries}\label{sec:u-connected_CQs}
We now turn to the unlabeled setting, whose classification is presented in
Table~\ref{tab:unlabeled}. We start with an intractability result 
which follows directly from the well-known intractability of query evaluation in
probabilistic databases~\cite{suciu2011probabilistic}:
\begin{proposition}[\cite{suciu2011probabilistic}]
  \label{prp:unlabeled_1wp_on_connected}
  The \phomu{\OWP}{\Connected} problem is \mbox{\#P-hard}.
\end{proposition}

\begin{inlineproof}
  Example~3.3 of~\cite{suciu2011probabilistic} states that the conjunctive
  query $\exists x\exists y\exists z\: U(x,y)\land U(y,z)$ is
  \#P-hard on TID instances. In other words, \phomu{\{\to\to\}}{\All} is
  \#P-hard, which implies the \#P-hardness of
  \phomu{\OWP}{\All}.
  We conclude using
  Lemma~\ref{lem:disconnected-instances}, which provides a
  PTIME (Turing) reduction\footnote{Note that it is usual to define \#P-hardness under
  Turing reductions rather than under Karp reductions, as \#P is a
  counting complexity class.} from the \phomu{\OWP}{\Connected} problem to the
  \phomu{\OWP}{\All} problem. 
\end{inlineproof}

Note that this \phomu{\OWP}{\Connected} problem
can be phrased in a very simple way: given an unlabeled connected 
probabilistic graph $(H,\pi)$ 
and a length $m$ as input (namely, that of the \OWP graph query),
compute the probability that $H$ contains a directed path of length~$m$.

This result suggests that, to obtain tractability, we need to restrict the
instance graphs. In fact, such tractability results were already obtained in the
previous sections.
In Section~\ref{sec:disconnected-CQs}, we proved
(Proposition~\ref{prp:u-all-DWT}) that 
\phomu{\All}{\DWT} has PTIME combined complexity. Similarly, in the previous
section, we proved
that \phoml{\Connected}{\TWP} is PTIME
(Proposition~\ref{prp:l-connected-2WP}), which means
\phomu{\Connected}{\TWP} is also PTIME.
This completes the analysis of the unlabeled case for \OWP{}, \TWP{} and \DWT{}
instances (see Table~\ref{tab:unlabeled}), so 
the only remaining case is that of
\PT{} instances.

\begin{table}[t]
  \caption{Tractability of \phomU{} in the connected case (Section~\ref{sec:u-connected_CQs})}

  \centering
\begin{minipage}{.63\linewidth}
\begin{tabular}{c|ccccc}
    $\downarrow$$G$\qquad $H$$\rightarrow$ & \OWP & \TWP & \DWT & \PT & \Connected \\
    \hline
    \OWP & & & & \vlineright{} &
    \hardcell\ref{prp:unlabeled_1wp_on_connected}\\
    \hhline{~|~~~|-|>{\hardline}-}\regline
    \TWP & & & \vlineright{} & \hardcell\ref{prp:u-2WP-PT} & \hardcell\\
    \hhline{~|~~~|->{\hardline}-}\regline
    \DWT & & & & \vlineright{\ref{prp:u-DWT-PT}} & \hardcell\\
    \hhline{~|~~~|-|>{\hardline}-}\regline
    \PT & & & \vlineright{} & \hardcell & \hardcell\\
    \Connected & &\ref{prp:l-connected-2WP} &
    \vlineright{\ref{prp:u-all-DWT}} & \hardcell & \hardcell\\
    \hhline{~|--->{\hardline}--}
\end{tabular}
\label{tab:unlabeled}

\tabexplanation
\end{minipage}
\end{table}

We start our study of \phomU for \PT{} instances with the simplest queries,
namely, \OWP{}, for which we will show tractability.
We will proceed by translating the \OWP{} query to a \emph{bottom-up deterministic tree
automata}~\cite{tata}:

\begin{definition}
	Given an alphabet $\Gamma$, a \emph{bottom-up
deterministic tree automaton}
on full binary (every node has either $0$ or $2$ children) rooted trees whose nodes are labeled by $\Gamma$ is a tuple $A = (Q,F,\iota,\Delta)$,
where:
        \smallskip
\begin{itemize}[(i)]
	\item $Q$ is a finite set of \emph{states};
	\item $F \in Q$ is a subset of \emph{accepting states};
	\item $\iota : \Gamma \to Q$ 
          is an \emph{initialization function} 
          determining the state of a leaf from its label;
	\item $\Delta :  \Gamma \times Q^2
          \to Q$ 
          is a \emph{transition function} 
		determining the state of an internal 
                node from its label and the states of its two
                children.
\end{itemize}
        \smallskip
Given a $\Gamma$-tree $\la T,\lambda\ra$ 
(where $\lambda : T \to \Gamma$ is the \emph{labeling function}),
we define the \emph{run} of $A$ on~$\la T,\lambda\ra$
as the function $\phi : T \to Q$ such that
(1)~$\phi(l) = \iota(\lambda(l))$  for every leaf $l$ of~$T$; and
(2)~$\phi(n) = \Delta( \lambda(n), \phi(n_1),\phi(n_2))$
for every internal node $n$ of~$T$ with children $n_1$ and $n_2$.
The automaton $A$ \emph{accepts} $\la T,\lambda\ra$
if its run on~$T$ maps
the root of~$T$ to a state of~$F$.
\end{definition}

We will evaluate \OWP{} queries by translating them to a tree automaton and
running it on an uncertain tree.
This will use again the notion of \emph{lineage} 
(recall Definition~\ref{def:lineage}),
which was extended in~\cite{amarilli2015provenance} to
tree automata running on trees with uncertain Boolean labels: the \emph{lineage}
of an automaton on such a tree
describes the set of annotations of the tree that makes the automaton accept. 
In this context, the lineage of
\emph{deterministic} tree automata 
was shown in~\cite{amarilli2016tractable} to be
compilable to a
\emph{deterministic decomposable negation normal form}
circuit~\cite{darwiche2001tractable}:

\begin{definition}
	A \emph{deterministic decomposable negation normal form} (d-DNNF) is a Boolean circuit $C$ with the following properties:

        \begin{enumerate}[(i)]
		\item negations are only applied to input gates;
		\item the inputs of any AND-gate depend on disjoint sets of input gates;
		\item the inputs of any OR-gate are mutually exclusive,
                  i.e., for any two input gates $g_1 \neq g_2$ of $g$, there is no valuation of the inputs of $C$
			under which $g_1$ and~$g_2$ both evaluate to true.
	\end{enumerate}

\end{definition}
We can then straightforwardly extend the Boolean probability computation problem (Definition~\ref{def:booleanproba}) to take circuits as inputs, and
the properties of \mbox{d-DNNF} circuits are designed to ensure that the Boolean probability computation problem restricted to d-DNNF has PTIME complexity~\cite{darwiche2001tractable}.
Combining these tools, we can show that \phomU on one-way path queries and polytree instances is tractable:

\begin{propositionrep}
	\label{prp:u-1WP-PT}
	\phomu{\OWP}{\PT} is PTIME.
\end{propositionrep}
\begin{proofsketch}
	The idea of the proof is to construct in polynomial time in the query
        graph~$G$ a bottom-up deterministic
        automaton~$A_G$, which runs on binary trees~$T$ representing 
	possible worlds of the polytree instance $H$, and accepts such a
        tree~$T$ iff the corresponding possible world satisfies $G$.
	We can then construct a d-DNNF representation of the lineage of $G$ on $H$
        by~\cite[Theorem~6.11]{amarilli2016tractable},
        which allows us to efficiently compute $\Pr(G \leadsto H)$:
        the complexity of this process is in $O(\left|A_G\right| \cdot \left|H\right|)$, hence polynomial
        in $\left|H\right| \cdot \left|G\right|$.
        (An alternative way to see this is to use the results of \cite{cohen2009running}.)

        Intuitively, the design of the bottom-up automaton ensures that, when it
        reaches a node~$n$ after having processed the subtree $T_n$ rooted
        at~$n$, its state reflects three linear-sized quantities about $T_n$:

        \begin{enumerate}
        \item the length of the longest path leading out of~$n$;
        \item the length of the longest path leading to~$n$;
        \item the length of the longest path overall in~$T_n$ (not necessarily
          via~$n$).
        \end{enumerate}

        The final states are those where the third quantity is greater than the
        length of~$G$. The transitions compute each triple from the child
        triples by considering how the longest leading paths are extended, and
        how longer overall paths can be formed by joining an incoming and
        outgoing path.
\end{proofsketch}
\begin{proof}
	Let $G$, $(H,\pi)$ be the \OWP{} query graph and the probabilistic \PT{} instance, and $m$ be the length of $G$. 
	Then $\Pr(G \leadsto H)$ is the probability that $H$ contains a directed path of length at least $m$.

	Because we will use automata that run on full binary trees, we will have to represent possible worlds of $H$ as full binary trees.
	The first step is to transform $H$ in linear time into a full binary polytree $H'$
        by applying a variant of the left-child-right-sibling encoding:
        in so doing, in addition to unlabeled edges of both orientations that exist in the polytree,
        we will also introduce some edges called \emph{$\epsilon$-edges} that are labeled by~$\epsilon$ and whose orientation does not matter
        (so we see them as undirected edges and write them $a - b$);
        intuitively, the $\epsilon$-edge $a - b$ means that $a$ and $b$ are in fact the same.
	For a node $a \in H$ and a child $b$ of $a$, we say that $b$ is an \emph{up-child} of $a$ if we have $b \rightarrow a$ and a \emph{down-child} of $a$
	if we have $a \rightarrow b$.
	We do this transformation by processing $H$ bottom-up as follows:

	\begin{itemize}
		\item If $n$ is a leaf node of $H$, then create a node $n'$ in~$H'$.
		\item If $n$ is an internal node of $H$ with up-children $u_1,\ldots,u_k$ and down-children $d_1,\ldots,d_l$ then,
			letting $u'_1,\ldots,u'_k$ and $d'_1,\ldots,d'_l$ be the corresponding nodes in $H'$:
			create a node $n'$ in~$H'$ and nodes $n'_1,\ldots,n'_{k+l-2}$ with the following $\epsilon$-edges:
			$n' - n'_1 - \ldots - n'_{k+l-2}$, all having probability $1$. 
			Create an edge $u'_1 \rightarrow n'$ whose probability
                        is that of $u_1 \rightarrow n$.
			For $2 \leq i \leq k$ create an edge $u'_{i} \rightarrow
                        n'_{i-1}$ annotated with the same probability as $u_{i} \rightarrow n$.
			For
                        $1 \leq i \leq l-1$
                        create an edge $n'_{k-1+i} \rightarrow d'_i$ annotated with the
                        same
                        probability as $n \rightarrow d'_i$, and finally
                        create an edge $n'_{k-2+l} \rightarrow d'_l$ annotated with the
                        same probability as $n \rightarrow d'_l$. Last, if any
                        node has exactly one children (specifically, $n'$, in
                        case $k + l = 1$), then create a node $n''$ in~$H'$ and
                        connect it with an $\epsilon$-edge to the node.
	\end{itemize}
	One can check that $H'$ is indeed a full binary polytree (with some
        edges being labeled by~$\epsilon$ and being undirected) and that 
	$\Pr(G \leadsto H)$ equals the probability that $H'$ contains a path of the form
        $(\rightarrow {-}^{*})^m$, that is, $m$ occurrences of a directed
        edge~$\rightarrow$ followed by some sequence of $\epsilon$-edges~$-$.

	The second step is to transform in linear time $H'$ into a probabilistic tree $T$ 
        to which we can apply the construction
        of~\cite{amarilli2015provenance_extended}.
        Specifically, $T$ must be 
	an ordered full binary rooted tree whose edges do not have a
        label or an orientation, but whose nodes $n$ carry a label in some
        finite alphabet $\Gamma$ (written $\lambda(n)$, where $\lambda$ is the
        labeling function) and with a
        probability value written $\pi(n)$.
        Writing  $\overline{\Gamma} \colonequals \Gamma \times \{0,1\}$
        as in \cite{amarilli2015provenance_extended},
        the semantics
        of~$T$ is that it stands for a probability distribution on
        $\bar{\Gamma}$-trees, i.e., trees $T'$ labeled with $\Gamma \times
        \{0, 1\}$,
        which have same skeleton as~$T$: for each node $n$ of~$T$, the
        corresponding node $n'$ in a possible world $T'$ has label $(\lambda(n),
        1)$ with probability $\pi(n)$ and label $(\lambda(n), 0)$
        otherwise. We do this transformation by first adding a new root vertex
        to~$H'$ with an $\epsilon$-edge with probability~$1$ to the original root
        (this clearly does not change the probability that $H'$ has a path of the
        prescribed form), and then simply create $T$ from $H'$ by assigning the
	label and probability of each node that is not the new root as the
        direction of its parent edge (in $\Gamma \defeq \{\uparrow, \downarrow, -\}$) and its probability (so the root of~$T'$ has label
        $-$ and probability~$1$).

	Our last step is to construct a bDTA $A_G$ running on
        $\overline{\Gamma}$-trees such that for every possible world $W$ of
        $H'$, letting $T_W$ be its representation as a $\bar{\Gamma}$-tree,
	$A_G$ accepts $T_W$ if and only if $W$ contains a path of the form
        $(\rightarrow {-}^{*})^m$.
	The states of~$A_G$ are of the form $\la \uparrow :~i, \downarrow:~j , \mathrm{Max:}~k \ra$ for
	$0 \leq i,j \leq k \leq m$, which ensures that $A_G$ is of size
        polynomial in~$|G|$ (and we will construct it in PTIME from $G$).
	The idea is that when a node $n$ of~$T_W$ will be in such a state, it will mean that:

	\begin{itemize}
		\item Letting $W_n$ be the subinstance of~$W$ which is
                  represented by the subtree of $T_W$ rooted at $n$, and letting $r_n$ be the root 
			of $W_n$, the longest directed upwards path in $W_n$ finishing at $r_n$ has length $i$
			(the path is the longest of the form $( \uparrow -^*)^*$ that ends at
			$r_n$).
		\item The longest directed downwards path in $W_n$ beginning at $r_n$ has length $j$
			(the path is the longest of the form $(\downarrow -^*)^*$ that begins at $r_n$).
		\item The longest directed path in $W_n$ has length $k$ (the path is of the form $(\rightarrow -^*)^k$ and is the longest in $W_n$).
	\end{itemize}
	We now describe the initialization function $\iota$ of $A_G$:

	\begin{itemize}
		\item $\iota((s,0)) \colonequals \la \uparrow :~0, \downarrow:~0 ,
                  \mathrm{Max:}~0 \ra$ for any $s \in \Gamma$.
		\item $\iota((-,1)) \colonequals \la \uparrow :~0, \downarrow:~0 ,
                  \mathrm{Max:}~0 \ra$.
		\item $\iota((\uparrow,1)) \colonequals \la \uparrow :~1, \downarrow:~0 ,
                  \mathrm{Max:}~1 \ra$.
		\item $\iota((\downarrow,1)) \colonequals \la \uparrow :~0, \downarrow:~1 ,
                  \mathrm{Max:}~1 \ra$.
		\item $\Delta((\uparrow,1) , \la \uparrow :~i, \downarrow:~j , \mathrm{Max:}~k \ra, \la \uparrow :~i', \downarrow:~j' , \mathrm{Max:}~k' \ra) 
			\colonequals \la \uparrow :~i'', \downarrow:~0 , \mathrm{Max:}~k''
                        \ra$ where $i'' \colonequals \min(m, \max(i+1,i'+1))$ and
			$k'' \colonequals \min(m,\max(i'', i+j', i'+j, k, k'))$. 
		\item $\Delta((\downarrow,1) , \la \uparrow :~i, \downarrow:~j , \mathrm{Max:}~k \ra, \la \uparrow :~i', \downarrow:~j' , \mathrm{Max:}~k' \ra) 
			\colonequals \la \uparrow :~0, \downarrow:~j'' , \mathrm{Max:}~k''
                        \ra$ where $j'' \colonequals \min(m, \max(j+1,j'+1))$ and
			$k'' \colonequals \min(m,\max(j'', i+j', i'+j, k, k'))$. 
		\item $\Delta((-,1) , \la \uparrow :~i, \downarrow:~j , \mathrm{Max:}~k \ra, \la \uparrow :~i', \downarrow:~j' , \mathrm{Max:}~k' \ra) 
			\colonequals \la \uparrow :~i'', \downarrow:~j'' ,
                        \mathrm{Max:}~k'' \ra$ where $i'' \colonequals
                        \max(i,i')$ and $j'' \colonequals \max(j,j')$
                        and $k'' \colonequals \min(m,\max( k, k', i+j', i'+j))$. 
		\item $\Delta((s,0) , \la \uparrow :~i, \downarrow:~j , \mathrm{Max:}~k \ra, \la \uparrow :~i', \downarrow:~j' , \mathrm{Max:}~k' \ra) 
			\colonequals \la \uparrow :~0, \downarrow:~0 , \mathrm{Max:}~k''
                        \ra$ where $k'' \colonequals \min(m, \max( k, k', i+j', i'+j))$ for every $s\in \{-,\uparrow,\downarrow\}$. 
	\end{itemize}
	The final state of $A_G$ are all the states $\la \uparrow :~i, \downarrow:~j , \mathrm{Max:}~k \ra$ such that $k = m$.
	One can check by a straightforward induction that the semantics of each
        state is respected, so that indeed the automaton tests the query~$G$.

	We conclude thanks to 
        Proposition~3.1 of
        \cite{amarilli2015provenance_extended} by computing in linear time in
        $|A_G|$ and $|H'|$ a representation of the lineage 
	on $H'$ of the query that checks whether the input contains a directed path of
	the form $(\rightarrow -^{*})^m$, and observe by Theorem~6.11
        of \cite{amarilli2016tractable} that it is a d-DNNF.
        We then compute the probability of this d-DNNF
        \cite{darwiche2001tractable}, yielding $\Pr(G \leadsto H)$ in PTIME: this
        concludes the proof.
\end{proof}

Hence, \PT{} instances enjoy tractability for the simplest query graphs. We now
study whether this result can be extended to more general queries. We first
notice that this result immediately extends to branching (i.e., to \DWT{}
queries), and even to unions of \DWT{} queries. Indeed, in the unlabeled
setting, as we already observed
at the end of Section~\ref{sec:labeled-disconnected-queries}, such queries are
equivalent to \OWP{} queries:

\begin{proposition}
	\label{prp:u-DWT-PT}
        \phomu{\DWT}{\PT} and \phomu{\UDWT}{\PT} are PTIME.
\end{proposition}
\begin{inlineproof}
        We first show the result for a \DWT{} query graph~$G$. Let $m$ be its
        height, i.e., the length of the longest
        directed path it contains, and 
	let $G'$ be the \OWP{} of length $m$, which can be computed in PTIME from
        $G$.
        It is easy to observe that $G$ and $G'$ are equivalent. Indeed, we can
        find $G'$ as a subgraph of~$G$ by taking any directed path of maximal
        length, and conversely there is
        a homomorphism from $G$ to $G'$: map the root of~$G$ to
        the first vertex of $G'$ and each element of~$G'$ at distance $i$ from
        the root to the $i$-th element of~$G'$.
	Hence, \phomU{} on $G$ and an input probabilistic \PT{} instance reduces to \phomU{} on $G'$ and
        the same instance, so that the result follows from Proposition~\ref{prp:u-1WP-PT}.

        The same argument extends to \UDWT{} by considering the greatest height
        of a connected component of~$G$.
\end{inlineproof}

Thus, we have successfully extended from \OWP{} to \UDWT{} queries
while preserving tractability on \PT{} instances. However, as we now show, tractability
is not preserved if we extend queries to allow two-wayness. Indeed:

\begin{proposition}
	\label{prp:u-2WP-PT}
	\phomu{\TWP}{\PT} is \#P-hard.
\end{proposition}
\begin{inlineproof}
	We adapt the proof of Proposition~\ref{prp:l-1WP-PT}, but we face the additional difficulty of not being allowed
	to use labels. Fortunately, we can use the two-wayness in the query graph
        to simulate labels.

        We reduce from \pptdnf{} (recall Definition~\ref{def:pp2dnf}):
        the input consists of two disjoint sets $\calX=\{X_1, \ldots,
        X_{n_1}\}$, $\calY=\{Y_1, \ldots, Y_{n_2}\}$ of Boolean
        variables, and a PP2DNF formula $\phi$.
        We construct a \TWP{} query graph~$G'$ and \PT{} instance~$H'$ with the
        same construction as the one that yielded $H$ and $G$ in that proof,
        except that we perform the following replacements (see
        Figure~\ref{fig:u-2WP-PT}):

        \begin{itemize}
          \item replace every edge $a \xrightarrow{S} b$ of~$H$ and $G$ by $3$ edges 
	$a \rightarrow \rightarrow \leftarrow b$;
          \item replace every edge $a \xrightarrow{T} b$ of~$H$ and $G$ by $3$ edges 
	$a \rightarrow \rightarrow \rightarrow b$.
        \end{itemize}

        In particular, the query graph is then defined as follows:
        \[G' \colonequals\,\,\rightarrow \rightarrow \rightarrow (\rightarrow
        \rightarrow \leftarrow)^{m+3} \rightarrow \rightarrow \rightarrow.\]
        All the edges of~$H'$ have probability $1$, except the middle edge of the edges that replaced the $S$-labeled edges
	used to code the valuation of the variables (e.g., for $X_i$, the middle edge of the 3 edges $X_i \rightarrow \rightarrow \leftarrow R$), 
	which have probability $\frac{1}{2}$.

	One can check that any image of  $G'$ 
	 must again go from the vertex $A_{x_j,j}$ to the vertex~$B_{y_j,j}$ 
         for some $1 \leq j \leq m$. 
	The key insight is that the first $\rightarrow^5$ of $G$ must be matched
        to a $\rightarrow^5$-path in $H'$, which only 
	exist as the concatenation of a $\rightarrow^3$ obtained by rewriting
        a $T$-edge for some variable~$X_j$, and of the first $\rightarrow^2$
        of the (undirected) path from $X_{x_j,j}$ to~$R$. Then there is no 
	choice left to match the next edges without failing.

	From this we deduce that,
        from any possible world $H_{\mathrm{W}}'$ of the modified instance $H'$,
        considering the corresponding possible world $H_{\mathrm{W}}$ of the
        unmodified instance~$H$ following the natural bijection,
        the modified query graph $G'$ has a homomorphism to~$H_{\mathrm{W}}'$ iff 
        the unmodified query graph $G$ has a homomorphism to~$H_{\mathrm{W}}$.
        We thus conclude that the probabilistic homomorphism problem on~$G'$ and~$H'$ has the same answer
        as the one on~$G$ and~$H$, which finishes the proof.
\end{inlineproof}

\begin{figure}
  \begin{tikzpicture}[xscale=1.52,yscale=.95,every node/.style={inner sep=.1,outer
    sep=.1}]
	\node[ name = H] at (0,5) {$H'$:};

	\node[ name = R, outer sep=3] at (3.5,5) {$R$};

	\node[ name = X1] at (0.5,3) {$X_1$};
	\node[ name = X2] at (3,3) {$X_2$};

	\node[ name = Y1] at (4,3) {$Y_1$};
	\node[ name = Y2] at (6.5,3) {$Y_2$};

	\node[ name = X13] at (0.5,1.5) {$X_{1,3}$};
	\node[ name = X12] at (0.5,0) {$X_{1,2}$};
	\node[ name = X11] at (0.5,-1.5) {$X_{1,1}$};

	\node[ name = X23] at (3,1.5) {$X_{2,3}$};
	\node[ name = X22] at (3,0) {$X_{2,2}$};
	\node[ name = X21] at (3,-1.5) {$X_{2,1}$};

	\node[ name = Y11] at (4,1.5) {$Y_{1,1}$};
	\node[ name = Y12] at (4,0) {$Y_{1,2}$};
	\node[ name = Y13] at (4,-1.5) {$Y_{1,3}$};

	\node[ name = Y21] at (6.5,1.5) {$Y_{2,1}$};
	\node[ name = Y22] at (6.5,0) {$Y_{2,2}$};
	\node[ name = Y23] at (6.5,-1.5) {$Y_{2,3}$};

	\node[ name = A11] at (-1,-2.5) {$A_{1,1}$};
	\node[ name = A12] at (-1,-1) {$A_{1,2}$};

	\node[ name = A23] at (1.5,0.5) {$A_{2,3}$};

	\node[ name = B12] at (5.5,-1) {$B_{1,2}$};

	\node[ name = B21] at (8,0.5) {$B_{2,1}$};
	\node[ name = B23] at (8,-2.5) {$B_{2,3}$};

	\draw[line width = 1pt, ->]  (X1) -- (1.5,3.666);
	\draw[dashed, line width = 1pt, ->]  (1.5, 3.666) -- (2.5, 4.333);
	\draw[line width = 1pt, <-]  (2.5,4.333) -- (R);

	\draw[line width = 1pt, ->]  (X2) -- (3.166, 3.666);
	\draw[dashed, line width = 1pt, ->]  (3.1666, 3.666) -- (3.333, 4.333);
	\draw[line width = 1pt, <-]  (3.333, 4.333) -- (R);

	\draw[line width = 1pt, ->]  (R) -- (3.666, 4.333);
	\draw[dashed, line width = 1pt, ->]  (3.666, 4.333) -- (3.833, 3.666);
	\draw[line width = 1pt, <-]  (3.833, 3.666) -- (Y1);

	\draw[line width = 1pt, ->]  (R) -- (4.5, 4.333);
	\draw[dashed, line width = 1pt, ->]  (4.5, 4.333) -- (5.5, 3.666);
	\draw[line width = 1pt, <-]  (5.5, 3.666) -- (Y2);

	\draw[line width = 1pt, ->]  (X1) -- (0.5, 2.45);
	\draw[line width = 1pt, <-]  (0.5, 2.45) -- (0.5, 2.05);
	\draw[line width = 1pt, <-]  (0.5, 2.05) -- (X13);

	\draw[line width = 1pt, ->]  (X13) -- (0.5, 0.95);
	\draw[line width = 1pt, <-]  (0.5, 0.95) -- (0.5, 0.55);
	\draw[line width = 1pt, <-]  (0.5, 0.55) -- (X12);

	\draw[line width = 1pt, ->]  (X12) -- (0.5, -0.55);
	\draw[line width = 1pt, <-]  (0.5, -0.55) -- (0.5, -0.95);
	\draw[line width = 1pt, <-]  (0.5, -0.95) -- (X11);

	\draw[line width = 1pt, ->]  (X2) -- (3, 2.45);
	\draw[line width = 1pt, <-]  (3, 2.45) -- (3, 2.05);
	\draw[line width = 1pt, <-]  (3, 2.05) -- (X23);

	\draw[line width = 1pt, ->]  (X23) -- (3, 0.95);
	\draw[line width = 1pt, <-]  (3, 0.95) -- (3, 0.55);
	\draw[line width = 1pt, <-]  (3, 0.55) -- (X22);

	\draw[line width = 1pt, ->]  (X22) -- (3, -0.55);
	\draw[line width = 1pt, <-]  (3, -0.55) -- (3, -0.95);
	\draw[line width = 1pt, <-]  (3, -0.95) -- (X21);

	\draw[line width = 1pt, ->]  (Y1) -- (4, 2.45);
	\draw[line width = 1pt, ->]  (4, 2.45) -- (4, 2.05);
	\draw[line width = 1pt, <-]  (4, 2.05) -- (Y11);

	\draw[line width = 1pt, ->]  (Y11) -- (4, 0.95);
	\draw[line width = 1pt, ->]  (4, 0.95) -- (4, 0.55);
	\draw[line width = 1pt, <-]  (4, 0.55) -- (Y12);

	\draw[line width = 1pt, ->]  (Y12) -- (4, -0.55);
	\draw[line width = 1pt, ->]  (4, -0.55) -- (4, -0.95);
	\draw[line width = 1pt, <-]  (4, -0.95) -- (Y13);

	\draw[line width = 1pt, ->]  (Y2) -- (6.5, 2.45);
	\draw[line width = 1pt, ->]  (6.5, 2.45) -- (6.5, 2.05);
	\draw[line width = 1pt, <-]  (6.5, 2.05) -- (Y21);

	\draw[line width = 1pt, ->]  (Y21) -- (6.5, 0.95);
	\draw[line width = 1pt, ->]  (6.5, 0.95) -- (6.5, 0.55);
	\draw[line width = 1pt, <-]  (6.5, 0.55) -- (Y22);

	\draw[line width = 1pt, ->]  (Y22) -- (6.5, -0.55);
	\draw[line width = 1pt, ->]  (6.5, -0.55) -- (6.5, -0.95);
	\draw[line width = 1pt, <-]  (6.5, -0.95) -- (Y23);

	\draw[line width = 1pt, ->]  (A11) -- (-0.4,-2.1);
	\draw[line width = 1pt, ->]  (-0.4,-2.1) -- (-0.1,-1.9);
	\draw[line width = 1pt, ->]  (-0.1,-1.9) -- (X11);

	\draw[line width = 1pt, ->]  (A12) -- (-0.4,-0.6);
	\draw[line width = 1pt, ->]  (-0.4,-0.6) -- (-0.1, -0.4);
	\draw[line width = 1pt, ->]  (-0.1,-0.4) --(X12);

	\draw[line width = 1pt, ->]  (A23) -- (2.1,0.9);
	\draw[line width = 1pt, ->]  (2.1,0.9) -- (2.4,1.1);
	\draw[line width = 1pt, ->]  (2.4,1.1) --(X23);

	\draw[line width = 1pt, ->]  (Y12) -- (4.6,-0.4);
	\draw[line width = 1pt, ->]  (4.6,-0.4) -- (4.9,-0.6);
	\draw[line width = 1pt, ->]  (4.9,-0.6) -- (B12);

	\draw[line width = 1pt, ->]  (Y21) -- (7.1,1.1);
	\draw[line width = 1pt, ->]  (7.1,1.1) -- (7.4,0.9);
	\draw[line width = 1pt, ->]  (7.4,0.9) -- (B21);

	\draw[line width = 1pt, ->]  (Y23) -- (7.1,-1.9);
	\draw[line width = 1pt, ->]  (7.1,-1.9) -- (7.4,-2.1);
	\draw[line width = 1pt, ->]  (7.4,-2.1) -- (B23);

	\node[fill = white, name = G] at (3.5,-3) {$G'$:   $\rightarrow \rightarrow \rightarrow (\rightarrow \rightarrow \leftarrow)^{6} \rightarrow \rightarrow \rightarrow$};

\end{tikzpicture}
\centering
\caption{Illustration of the proof of Proposition~\ref{prp:u-2WP-PT} for the
PP2DNF formula $X_1 Y_2 \lor X_1 Y_1 \lor X_2 Y_2$. Dashed edges have probability $\frac{1}{2}$, all others have probability~$1$.}
\label{fig:u-2WP-PT}
\end{figure}
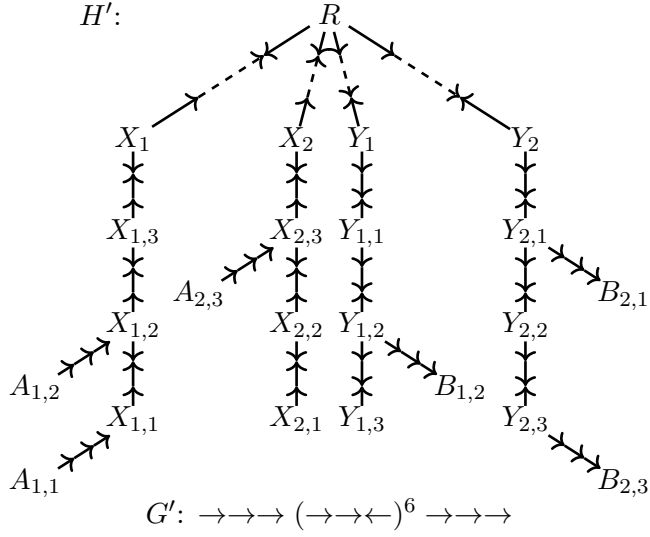

\section{Conclusion}\label{sec:conclusion}
We have introduced the probabilistic homomorphism problem,
also known in the database community as probabilistic evaluation of conjunctive queries
on TID instances, and studied its combined complexity for 
various restricted classes of query and instance graphs.
Our classes illustrate the impact on \phom of
various features: acyclicity, two-wayness, branching, connectedness, and labeling.
As we show, the landscape is already quite enigmatic, even for those seemingly restricted classes!
In particular, we have identified four incomparable maximal
tractable cases, reflecting various tradeoffs between the expressiveness that we
can allow in the queries and in the instances:
\begin{itemize}
\item arbitrary queries on unlabeled downward trees
  (Proposition~\ref{prp:u-all-DWT});
\item one-way path queries on labeled downward trees
  (Proposition~\ref{prp:l-1WP-DWT});
\item connected queries on two-way labeled path instances
  (Proposition~\ref{prp:l-connected-2WP});
\item downward tree queries on unlabeled polytrees
  (Pro\-position~\ref{prp:u-DWT-PT});
\end{itemize}
These results all extend to disconnected instances, as shown in
Section~\ref{sec:disconnected-instances}.
The
(somewhat sinuous)
tractability border is described in
Tables~\ref{tab:unlabeled-conj_of},~\ref{tab:labeled}, and~\ref{tab:unlabeled}.

It is debatable whether the tractable classes we have identified yield
interesting tractable cases for practical applications. The
settings of Propositions~\ref{prp:u-all-DWT}, \ref{prp:u-DWT-PT},
and~\ref{prp:l-connected-2WP} may look restrictive, as both labels and branching
are important features of real-world instances, though some situations
may involve unlabeled tree-like instances, or labeled words.
The setting of Proposition~\ref{prp:l-1WP-DWT} may be richer,
and is reminiscent of probabilistic
XML~\cite{kimelfeld2013probabilistic}: the instance is a labeled
(downward) tree, while the query is a path evaluated on that
tree.

\paragraph*{Future work.}
The query and instance features studied in this paper could be completed by
other dimensions: e.g., studying an \emph{unweighted} case inspired by counting
CSP where all probabilities are $1/2$ (as our hardness proofs seem to
heavily rely
on some edges being certain); imposing \emph{symmetry} in the
sense of \cite{beame2015symmetric}; or alternatively restricting the
\emph{degree} of graphs (though all our hardness proofs on polytrees and lower
classes can seemingly be modified to work on bounded-degree). Another option
would be to modify some of the existing dimensions: first, 
polytrees could be generalized to \emph{bounded-treewidth instances}, as we believe
that the relevant tractability result (Proposition~\ref{prp:u-DWT-PT}) adapts to
this setting; second, non-branching instances could be generalized to
\emph{bounded-pathwidth instances}, or maybe general instances with the
$\underline{X}$-property (recall Definition~\ref{def:X-property}).

Of course, another natural direction would be to lift the arity-two restriction,
although it is not immediate to generalize the definition of our classes to
work in higher arity signatures.
We could also extend the query language: in particular, allow \emph{unions}
of conjunctive queries as in~\cite{dalvi2012dichotomy}; allow a
\emph{descendant} axis in the spirit of XML query languages~\cite{benedikt2009xpath}; or more generally allow fixpoint constructs as done
in~\cite{amarilli2017combined} in the non-probabilistic case. An interesting
question is whether an extended query language could capture the
tractability results obtained in the context of probabilistic XML
by~\cite{cohen2009running} (remembering, however, that such
results crucially depend on having an order relation on node children~\cite{amarilli2014possibility}).
Another possibility would be to search for extensions of the
\mbox{$\beta$-acyclicity}
approach, and investigate which restrictions on the queries and instances ensure
that the lineages are \mbox{$\beta$-acyclic}.

Last, the connection to CSP would seem to warrant further investigation. In
particular, we do not know whether one could show a general dichotomy result on
the combined complexity of query evaluation on TID instances, to provide a
probabilistic analogue to the Feder--Vardi conjecture~\cite{feder1998computational}.

\paragraph*{Acknowledgements.}
We are grateful to Florent Capelli for pointing out to us the connection to
$\beta$-acyclic instances and suggesting the idea used in the proof of
Proposition~\ref{prp:l-1WP-DWT}, and to Tyson Williams for pointing out a connection to
holographic algorithms,
in a very informative CSTheory post~\cite{cstheory_tyson_williams}.
We also thank anonymous referees for their valuable feedback.
This work was partly funded by the Télécom ParisTech Research Chair on Big Data
and Market Insights.

\begin{toappendix}
  \section{Proof of Hardness of Counting Edge Covers}\label{apx:holographic}
  In this appendix, following a connection pointed out
in~\cite{cstheory_tyson_williams}, we give a proof of the following
strengthening of Theorem~\ref{thm:bipec}, which is independent from the proof
of~\cite{khanna2011queries}.

\begin{theorem}
  \label{thm:bipec2}
  The \bipec{} problem is \#P-complete. Hardness holds even for 2--3
  regular bipartite undirected graphs that are planar.
\end{theorem}

\begin{proof}
  Membership in~\#P is straightforward: the machine guesses a subset of edges
  and accepts in PTIME iff the subset is a matching. Hence, we focus on hardness.

        Recall that \emph{$2$--$3$ regular bipartite undirected graphs} are
        bipartite undirected graphs $\Gamma = (U \sqcup V, E)$ where the degree
of each vertex in $U$ is $2$ and that of each vertex in $V$ is~$3$.
  We will show how the result derives from the
  holographic reduction results of~\cite{cai2012holographic}.

  For $t \in U \sqcup V$, we denote by $E(t)$ the set of edges to which $t$ is
  adjacent.
For a valuation of the edges $\nu: E \to \{0,1\}$ and a vertex $t$, we write $\nu(E(t))$ the multiset $\{\!\!\{ \nu(e) \mid e \in E(t) \}\!\!\}$.
Given a multiset of bits $B$, the \emph{Hamming weight} of~$B$ is the number of $1$ bits in~$B$.
For each $x_0, \ldots, x_n \in \{0, 1\}$, let $[x_0, \ldots, x_n]$ denote
  the function that takes a multiset of~$n$ bits as input and outputs $x_i$ if the Hamming weight
of those $n$ bits is $i$.

  For every $x_0, x_1, x_2, y_0, y_1, y_2, y_3 \in \{0,1\}$, the problem
  \#$[x_0, x_1, x_2] | [y_0, y_1, y_2, y_3]$ is the following~\cite{cai2012holographic}:
  given a~\mbox{$2$--$3$} regular bipartite undirected graph $\Gamma = (U \sqcup V, E)$, 
compute the quantity 
  \[
    \sum_{\nu: E \to \{0,1\}} \prod_{u \in U} [x_0, x_1, x_2] (\nu(E(u))
  \times\prod_{v \in V} [y_0, y_1, y_2, y_3] ( \nu(E(v))  ).
  \]

  Then, when we restrict our attention to $2$--$3$ regular bipartite undirected graphs, our problem
  \bipec{} can be seen to be the same as $\#[0, 1, 1] | [0, 1, 1, 1]$. 
Indeed, seeing a valuation~$\nu$ of the edges as a set of edges, the
  value under the sum for a valuation $\nu$ will be $1$ if and only if,
  for every vertex, there exists
an adjacent edge such that $\nu(e) = 1$, which exactly means that $\nu$ is an edge cover of $\Gamma$.

Now, let us consider the problem $\#[1, 1, 0] | [1, 1, 1, 0]$.
As observed at the end of Section~8
  of~\cite{cai2008holographic}, it is the
    \emph{reversal} of the problem $\#[0, 1, 1] | [0, 1, 1, 1]$.
    Indeed, the problem $\#[1, 1, 0] | [1, 1, 1, 0]$ amounts
to counting the number of subsets $S$ of edges such that, for every
vertex $v$, there exists at least one edge adjacent to~$v$ that is not
in~$S$, i.e., that is in~$E\setminus S$. But this means that $\#[1, 1, 0] |
[1, 1, 1, 0]$ counts the number of sets $S$ such that $E \setminus S$ is
an edge cover of~$\Gamma$. As there is a trivial bijection between the sets
$S$ that are edge covers and the sets $S'$ such that $E \setminus S'$ is
an edge cover,
\bipec{} is PTIME-equivalent to $\#[1, 1, 0] | [1, 1, 1, 0]$ on $2$--$3$
regular bipartite undirected graphs.

  Now, the problem $\#[1,1,0] | [1,1,1,0]$ is shown in~\cite{cai2008holographic,cai2012holographic} to be \#P-hard.
However, one subtle problem is that the notion of graph used in these works is
different from our own notion, because they generally allow the graph
  to contain multiple occurrences of the same edge, and they allow \emph{self-loops}
(edges from a vertex to itself): this is what we would call a
\emph{multigraph}. By contrast, the graphs that we used throughout the paper are
\emph{simple graphs}, where edges cannot be repeated, and we disallow
self-loops.
Now, even though this is implicit in~\cite{cai2008holographic,cai2012holographic}, the line of works on Holant problems typically consider
multigraphs instead of graphs (e.g., see footnote~1 page~217
of~\cite{williams2015advances}), so the hardness result
in~\cite{cai2008holographic,cai2012holographic} would in principle apply to
input \emph{multigraphs}. Thus, it does not directly allow us to conclude the proof of
our result, as we stated that the problem \bipec{} is \#P-hard on
\emph{simple} graphs.
For this reason, to conclude our proof, we will follow in detail the reasoning used by~\cite{cai2008holographic} to
show that $\#[1,1,0] | [1,1,1,0]$ is \#P-hard on 2--3 regular bipartite planar
(multi)graphs, and check that their argument in fact establishes hardness even
when the input graphs are required to be simple.

The first step of the hardness proof
for $\#[1,1,0] | [1,1,1,0]$ 
in~\cite{cai2008holographic} is \cite[Lemma 4.1]{cai2008holographic}. For our
purposes, this lemma shows that the problem $\#[x_0,x_1,x_2] | [1,1,1,0]$ is
\#P-hard for some~$x_0,x_1,x_2 \in \mathbb{R}$ (which uses some generalized
definition of the notation \#$[x_0, x_1, x_2] | [y_0, y_1, y_2, y_3]$ that
allows
values outside of $\{0, 1\}$), and that this hardness holds even for inputs that are 2--3 regular
bipartite planar (multi)graphs. We want to show that this is true also when the input
graphs are required to be simple. 
This hardness result in \cite[Lemma 4.1]{cai2008holographic} is shown by reducing from the problem of counting vertex covers on~$3$-regular planar (multi)graphs, which
is \#P-hard by~\cite{xia2007computational}: as they explain, this implies that
the problem $\#[0,1,1] | [1,0,0,1]$ is \#P-hard on 2--3 regular bipartite planar
(multi)graphs. Now, we note that \cite[Theorem~9]{xia2007computational} actually shows
hardness of the problem \#3RBP-VC of counting vertex covers
on $3$-regular planar (multi)graphs that are \emph{bipartite}: this means in
particular that these input graphs cannot contain self-loops (but they can contain
multi-edges). We can use this to show \#P-hardness of the problem
$\#[0,1,1] | [1,0,0,1]$ on 2--3 regular bipartite \emph{simple} graphs. Indeed,
as in \cite[Lemma 4.1]{cai2008holographic},
we reduce from \#3RBP-VC by taking an input graph $G$ to \#3RBP-VC and
creating a new graph $G'$ by subdividing each edge (adding
a vertex of degree~2 in the middle of the edge). Now, the resulting graph~$G'$ is
2--3 regular bipartite, it is still planar, it contains no multi-edges because $G$
does not contain any self-loops (any multi-edges in~$G$ get
translated to different edges in~$G'$ as each of them is subdivided with a different
middle vertex), and the answer to \#3RBP-VC on~$G$ is exactly the answer to $\#[0,1,1] |
[1,0,0,1]$ on~$G'$. Indeed, as explained in \cite[Lemma
4.1]{cai2008holographic}, counting the vertex covers on~$G$ amounts to
counting the number of subsets $S$ of the vertices of degree~3 of~$G'$
such that every vertex of degree~2 is adjacent to a vertex of that subset.
Equivalently, this is counting subsets $S'$ of the edges of~$G'$ such that each vertex of
degree~2 is adjacent to an edge of~$S'$ (hence the constraint $[0,1,1]$) and
such that for every vertex of degree~3 we either keep all incident edges in~$S'$ (i.e.,
we keep the vertex in~$S$) or keep none (i.e., we do not keep the vertex
in~$S$), hence the constraint $[1,0,0,1]$. Thus, this reduction shows that the
problem $\#[0,1,1] | [1,0,0,1]$ on~2--3 regular bipartite planar
graphs is \#P-hard, even when the input graphs are required to be simple. Let us
denote this problem by ($\star$).

The rest of the proof of~\cite[Lemma 4.1]{cai2008holographic} uses a holographic
reduction to reduce problem~($\star$) to the problem~$\#[x_0,x_1,x_2] | [1,1,1,0]$,
for some~$x_0,x_1,x_2 \in \mathbb{R}$. More precisely, the hardness is shown
by reducing either from problem~($\star$), or from a different problem of counting matchings,
depending on the signature that we are interested in. Specifically, the
reduction from~($\star$) is used when
reducing to a signature $\#[x_0,x_1,x_2] | [y_0,y_1,y_2,y_3]$ when we can find values
$\alpha_1, \alpha_2, \beta_1, \beta_2$ such that $\alpha_1 \beta_2 - \alpha_2
\beta_1 \neq 0$ and we have $y_i = \alpha_1^{3-i} \alpha_2^i +
\beta_1^{3-i}\beta_2^i$ for all~$i$. For the signature that we are interested in, namely
$[y_0,y_1,y_2,y_3] = [1,1,1,0]$, we can
take~$\alpha_1\colonequals 0,\alpha_2\colonequals -1,\beta_1\colonequals
1,\beta_2\colonequals 1$ and verify that the conditions
are satisfied, so indeed the reduction is from problem~($\star$). The reduction from
problem~($\star$) in \cite[Lemma 4.1]{cai2008holographic} is a holographic reduction,
which does not modify the input graphs (see for instance~\cite[Section 3.2]{williams2015advances} for an introduction to 
holographic reductions). Thus, it follows from the proof of~\cite[Lemma 4.1]{cai2008holographic} that
there is a choice of $x_0,x_1,x_2 \in \mathbb{R}$ such that the
problem~$\#[x_0,x_1,x_2] | [1,1,1,0]$ is \#P-hard on 2--3 bipartite planar
graphs that are simple.

Then, the authors of~\cite{cai2008holographic} use what is called the \emph{interpolation technique},
in order to show that any problem of the form~$\#[x_0,x_1,x_2] | [1,1,1,0]$ is
polynomial-time reducible (under Turing reductions) to the problem~$\#[1,1,0] |
[1,1,1,0]$, which is the problem that we are interested in. This is done in the
appendix of~\cite{cai2008holographic}, on page~15, by replacing the degree-2 nodes of the input graph by some graph gadgets.
On page 15 we see that they use graph gadget number 1, and by looking at that
gadget (Figure~3), we see that performing replacements with this gadget can
never introduce parallel edges or self-loops. Thus, as the reduction is from
simple graphs (as we argued in the preceding paragraphs), the graphs that are
images of this reduction are also simple graphs. Thus, the proof
of~\cite{cai2008holographic} actually shows that the problem~$\#[1,1,0] |
[1,1,1,0]$ is \#P-hard on 2--3 regular bipartite planar graphs that are simple.
This is the result needed to conclude the proof.
\end{proof}

\end{toappendix}

\bibliographystyle{abbrv}
\bibliography{main}

\end{document}